\documentclass[12pt]{article}
\usepackage[utf8]{inputenc}
\usepackage{graphicx, setspace,natbib}
\usepackage{xcolor}
\usepackage[a4paper, top=2.5cm, bottom=2.5cm, left=3cm, right=3cm]{geometry}
\usepackage{amsmath,amssymb,amsthm}
\usepackage{bm}
\usepackage{booktabs}
\usepackage{float}
\usepackage{setspace}
\usepackage[hidelinks]{hyperref}
\usepackage{caption} 
\def\equationautorefname~#1\null{Eq.(#1)\null}
\usepackage{microtype} 
\newtheorem{assumption}{Assumption}
\newtheorem{lemma}{Lemma}
\newtheorem{theorem}{Theorem}
\newtheorem{corollary}[theorem]{Corollary}

\begin{document}
\title{\Large Soft-Noncrossing Bayesian Panel Quantile Regression for Measuring Climate Tail Risk}

\author{Florian Huber \\
 \emph{University of Salzburg} \and Aubrey Poon \\
 \emph{University of Kent} \and Dan Zhu\\
 \emph{Monash University}}

\onehalfspacing
\maketitle
\begin{abstract}

\noindent We develop a hierarchical Bayesian panel quantile regression model in which unit-specific coefficient paths are smoothed across quantiles by Gaussian processes, while a common time effect absorbs aggregate shocks. Componentwise-monotone Bernstein polynomials, perturbed by unit-specific deviations, deliver soft noncrossing, and we provide identification conditions together with a bound on the crossing probability. Applying the model to 33 countries over 1979--2023, we find that global temperature shocks generate a systemic, non-diversifiable downside risk to output growth. This risk is concentrated in the lower tail and disproportionately affects emerging markets. Finally, we apply our framework to risk analysis and show that the model reduces out-of-sample tail-risk forecast loss by roughly one-third relative to country-specific quantile regressions.
\end{abstract}
\section{Introduction}

Many of the most consequential questions in macroeconomics and finance concern the tails of the outcome distribution rather than its center: the risk of a severe contraction in output, the vulnerability of growth to financial stress, and the downside exposure of an economy to aggregate disturbances \citep{adrian2019vulnerable,adrian2022term}. Quantile regression \citep{koenker1978regression,koenker2005quantile} is the natural instrument to apply to answer such questions, since it characterizes how covariates reshape the entire conditional distribution of an outcome rather than only its mean, and does so under weak assumptions on the data generating process \citep{angrist2006quantile}. Yet the questions that make quantile methods indispensable increasingly arise in panel settings where the object of interest is not a single conditional quantile but an entire cross section of unit-specific, quantile-indexed responses. In these settings standard quantile regression confronts three obstacles that existing methods address only in isolation.

First, quantile regressions are often estimated separately at each quantile level, which can produce noisy coefficient estimates and crossing quantile curves, especially when the time dimension is short. The problem of noncrossing quantile estimation has therefore received considerable attention in the quantile regression literature \citep{bondell2010noncrossing,santos2020noncrossing}, though existing solutions are largely frequentist and tailored to a single cross-section rather than a panel. Second, panel data typically contain both unit-specific heterogeneity and common time-varying shocks. Ignoring such common shocks may confound cross-sectional heterogeneity with aggregate dynamics, while treating time effects independently across quantiles can quickly lead to a high-dimensional and weakly identified model. Third, allowing slope coefficients to vary across both units and quantiles creates a large number of parameters, making regularization or partial pooling essential in moderate-sized panels.

This paper develops a hierarchical Bayesian panel quantile regression model that addresses these issues. The model allows unit-specific intercepts and slopes to vary smoothly across quantiles, while introducing a common dynamic time effect that is invariant across quantile levels. The quantile-invariant time effect captures aggregate shocks that shift the conditional distribution in parallel, such as global macroeconomic conditions, policy disturbances, financial shocks, or climate-related common factors. Unit-specific intercepts and slopes capture heterogeneous distributional responses across countries, firms, regions, or individuals. This separation makes it possible to distinguish common aggregate variation from heterogeneous unit-level responses.

To link information across quantiles, the unit-specific coefficient paths are modeled as stochastic processes indexed by the quantile level. In particular, we place a multivariate Gaussian process prior on the intercept and slope functions. This prior induces smoothness across nearby quantiles and provides partial pooling across units, improving estimation when the panel is short or when the number of quantile levels is large. This idea is related to recent approaches that smooth quantile functions or quantile coefficients across quantile levels \citep{li2025spline}, but differs by combining smooth coefficient paths with heterogeneous panel slopes and a dynamic common time effect. For Bayesian estimation, we exploit the connection between quantile regression and the asymmetric Laplace likelihood \citep{kozumi2011gibbs}.

A central difficulty in joint quantile modeling is the possibility of quantile crossing. Because conditional quantile functions must be nondecreasing in the quantile index, unrestricted quantile-specific estimates may violate a basic property of conditional distributions. We address this issue by imposing shape restrictions on the population coefficient paths. After transforming covariates to a bounded nonnegative support, componentwise monotonicity of the intercept and slope paths provides a sufficient condition for noncrossing over the transformed covariate space. We implement these restrictions using a Bernstein-polynomial representation of the population path. The unit-specific deviations from this monotone population path are modeled as smooth Gaussian-process deviations, yielding a form of soft noncrossing: individual quantile curves are shrunk toward a globally monotone population curve while still allowing data-driven heterogeneity across units.

The paper makes three contributions. First, we introduce a panel quantile regression model in which a quantile-invariant common component acts as a latent factor that absorbs aggregate shocks, cleanly separating systemic distributional variation from heterogeneous, unit-specific responses. This separation is essential whenever unobserved common shocks are present but the object of interest is a distribution rather than an average, as in growth-at-risk and related macro-financial applications. Second, we place a multivariate Gaussian process prior on the quantile-indexed coefficient paths, which borrows information smoothly across adjacent quantiles and partially pools across units. Relative to methods that smooth quantile functions within a single cross-section \citep{li2025spline}, the prior operates jointly across the quantile index and the cross-section and remains tractable in short panels. Third, we establish identification and prove a probability bound showing that quantile crossing becomes arbitrarily unlikely as the unit-specific deviations concentrate around a Bernstein-monotone population path --- turning noncrossing from a constraint imposed after estimation into a property the prior delivers by construction.

We apply our proposed panel quantile regression framework to a quarterly panel of 33 countries, drawn from \citet{mohaddes2024compilation}, to investigate the relationship between climate-related shocks and downside macroeconomic risk over the period 1979--2023. Each country's climate-related shocks are constructed following the recent methodology of \citet{10.1093/qje/qjag011}. The application is designed to illustrate what the framework delivers. We decompose climate disturbances into local versus global and level versus volatility components, and by tracing each across the conditional distribution of growth we isolate a systemic, non-diversifiable source of downside risk that neither a country-by-country nor a mean-based analysis can unveil.

We find that local temperature shocks generate modest and broadly symmetric responses across the conditional distribution of GDP growth at the panel level, masking substantial heterogeneity across countries that becomes apparent in our cross-country analysis. This is consistent with \citet{berg2024gdp}, who similarly document considerable heterogeneity in the sign and magnitude of GDP responses to idiosyncratic temperature shocks across countries, with the direction of the response varying systematically with country characteristics such as income, educational attainment, and trade openness. By contrast, global temperature shocks generate more pronounced negative effects, particularly in the lower tail of the conditional distribution, consistent with \citet{kahn2021long}, who show that persistent increases in temperature above historical norms are associated with significant reductions in long-run economic growth. Turning to volatility, local temperature volatility shocks produce a modest response skewed toward the lower tail, with the 90--10 per cent band widening as the horizon progresses. Global temperature volatility shocks, however, display a markedly different pattern, eliciting a positive response in output growth in the short term that dissipates over the medium to long term.

For all 33 countries, we construct a climate growth-at-risk measure that isolates the combined impact of the four climate shocks on each country's downside risk. We find that the impact of climate shocks on growth-at-risk is widespread and uniformly negative at the one-quarter horizon, with all 33 countries recording a deterioration in expected shortfall and a panel average of $-$0.460 percentage points. Emerging market economies bear a larger short-run climate tail risk than advanced economies, with EM and AE averages of $-$0.530 and $-$0.395 percentage points, respectively, and climate tail risks dissipate substantially at medium and longer horizons.

To evaluate the practical utility of our proposed framework, we conduct a pseudo-out-of-sample forecasting exercise, comparing one-quarter-ahead tail-risk forecasts from our panel quantile model against a country-specific quantile benchmark. We find strong and consistent evidence that the panel quantile regression produces more accurate growth-at-risk forecasts, reducing tail forecast loss by approximately 34 per cent relative to the country-specific benchmark.

The remainder of the paper is organized as follows. Section~\ref{sec: econometric_framework} introduces the panel quantile regression model, discusses identification, presents the monotonicity conditions used to control quantile crossing, and describes the Bayesian estimation strategy. Section~\ref{sec:climate} applies the model to a panel of 33 countries to study the distributional effects of temperature shocks and temperature volatility. Section~\ref{subsec:climate_robustness} carries out three robustness checks, and Section~\ref{sec:conclusion} concludes.

\section{Econometric Framework}\label{sec: econometric_framework}
\subsection{Panel data and Standard Two Way Fixed Effect Model}
Panel data consist of repeated observations on the same units (e.g.\ individuals, firms, regions) over multiple time periods. Formally, let $y_{it}\in \mathbb{R}$ denote the outcome of unit $i$ at time $t$,
\begin{equation}
y_{it} = a_i + f_t + \bm{x}_{it}'\bm{\beta}_i + v_{it},    \label{eq: panel}
\end{equation}
where $a_i$ and $f_t$ are unit and time-fixed effects, and $\bm{x}_{it} \in \mathbb{R}^{p-1}$ is a vector of covariates that impacts $y_{it}$ through unit-specific coefficients $\bm{\beta}_i$. We impose no parametric or distributional assumptions on the shocks $v_{it}$; identification instead rests on the conditional quantile restriction introduced below.

This two-way fixed-effects specification is now standard in applied work, as it simultaneously controls for both unit-specific heterogeneity and aggregate time shocks \citep{arellano2001panel}. Unlike the canonical two-way fixed-effects model, we allow the slopes to be unit-specific: the coefficient $\bm{\beta}_i$ represents how each unit translates changes in $\bm{x}_{it}$ into changes in the outcome. Such heterogeneous-coefficient panels move beyond a single ``average effect'' and provide a richer description of cross-sectional variation \citep{pesaran1995estimating}. The additive time effect $f_t$ is, moreover, the leading special case of the common-factor error structures used to capture cross-sectional dependence in large heterogeneous panels \citep{pesaran2006estimation}.

\subsection{Panel Quantile Regression with Heterogeneous Slopes}

The mean regression framework in \autoref{eq: panel} focuses on average effects. In many applications, however, researchers are interested in how covariates influence the entire conditional distribution of the outcome. For instance, the impact of education on wages may differ between low- and high-income individuals {\citep{buchinsky1994changes,angrist2006quantile}}. Or in time series settings where financial conditions trigger a stronger effect in the left tail of the distribution of output growth \citep{adrian2019vulnerable}. Quantile regression is a suitable tool to answer such questions.

Quantile regressions allow the parameters of the panel model to vary with the quantile index $\tau \in (0,1)$. Formally, the $\tau^{\text{th}}$ conditional quantile of $y_{it}$ given the information set $\mathcal{F}_t$ is specified as
\[
Q_{\tau}(y_{it}\mid \mathcal{F}_t) = a_i(\tau) + f_t + \bm{x}_{it}'\bm{\beta}_i(\tau),
\]
where $a_i(\tau)$ captures unit-specific heterogeneity at quantile $\tau$, $f_t$ represents common time effects (assumed invariant across $\tau$ for parsimony), and $\bm{\beta}_i(\tau)$ allows for heterogeneous slope coefficients across units and across quantiles. We choose to maintain the assumption that $f_t$ is invariant across $\tau$. The motivation is that $f_t$ captures unobserved aggregate shocks or macro-level influences that are common to all units in period $t$. These shocks play a role similar to observed covariates $\bm{x}_{it}$ that enter the model linearly and are assumed to have the same realization regardless of the quantile being estimated. In other words, just as $\bm{x}_{it}$ does not change with $\tau$, we interpret $f_t$ as an unobserved regressor that is likewise \emph{constant across quantiles} but time-varying. For the time effects $\{f_t\}$, we adopt a dynamic specification by modeling them as a latent AR(1) process with observed aggregate predictors. Specifically,
\[
f_t = \gamma f_{t-1} +\bm{\xi}'\bm{w}_t+ \eta_t, \qquad \eta_t \sim \mathcal{N}(0,\sigma_f^2),
\]
is not just a static time effect but evolves over time in response to global shocks $\bm{w}_t$, with corresponding coefficients $\bm{\xi}$. The dynamic factor $f_t$ captures unobserved common forces --- such as macroeconomic conditions, policy shifts, or broad financial shocks --- that affect all units simultaneously, while $\bm{w}_t$ represents observable indicators of such global disturbances. Importantly, $f_t$ is assumed invariant across quantiles, meaning these common shocks shift the conditional distribution of $y_{it}$ in parallel rather than altering its shape across quantiles. This ensures that the role of aggregate dynamics is separated from the heterogeneous unit-level responses embedded in $a_i(\tau)$ and $\bm{\beta}_i(\tau)$. The result is a flexible structure that disentangles individual heterogeneity from evolving global shocks, while preserving tractability within a Bayesian hierarchical framework.

It is useful to state the conditions under which the quantile-specific parameters are identified from the panel quantile restriction. Identification requires separating three distinct sources of variation: unit-specific heterogeneity, common time effects, and unit-specific responses to the regressors. Since the common time effect is assumed to be invariant across quantiles, this restriction must be imposed explicitly. In addition, a normalization is needed to separate the unit intercepts from the time effects, and a full-rank condition is required to rule out perfect collinearity among unit effects, time effects, and unit-specific regressors.
\begin{assumption}
\label{ass:quantile_identification}
Fix $\tau\in(0,1)$ and let $\bm{x}_{it}\in\mathbb{R}^{K}$, where $K=p-1$. The conditional $\tau$-quantile of $y_{it}$ is given by
\[
Q_{\tau}(y_{it}\mid \mathcal I_{it}) = a_i(\tau)+f_t+\bm{x}_{it}'\bm{\beta}_i(\tau),
\]
where $\mathcal I_{it}$ denotes the information set with respect to which the conditional quantile is defined. The following conditions hold.

\begin{enumerate}
\item     The common time effect $\{f_t\}_{t=1}^T$ is invariant across  quantiles. That is, the same sequence $\{f_t\}_{t=1}^T$ enters the conditional quantile function for every $\tau\in(0,1)$.

\item     The time effects satisfy the normalization
    \[
    \sum_{t=1}^T f_t = 0.
    \]
\item No linear combination of the unit effects, normalized time effects, and unit-specific regressors is perfectly collinear. Equivalently, after imposing $\sum_{t=1}^T f_t=0$, the stacked panel design containing unit effects, normalized time effects, and unit-specific regressors has full column rank.
\item   Define the quantile residual
    \[
    v_{it}(\tau)
    =
    y_{it}
    -
    a_i(\tau)
    -
    f_t
    -
    \bm{x}_{it}'\bm{\beta}_i(\tau).
    \]
Then
    \[
    P\left(v_{it}(\tau)\le 0\mid \mathcal I_{it}\right)=\tau.
    \]
\end{enumerate}
\end{assumption}

A few remarks on Assumption~\ref{ass:quantile_identification} are in order. First, the information set $\mathcal I_{it}$ is defined as the $\sigma$-algebra generated by the unit-specific intercepts and slopes $\{\alpha_i(\tau),\bm{\beta}_i(\tau)\}$, the common time effects $\{f_t\}$, and the observed covariates $\{\bm{x}_{it}\}$. The conditional quantile restriction in part~(d) therefore requires that, at the true parameter values, $\bm{x}_{it}$ and $f_t$ capture all systematic variation in the $\tau^{\text{th}}$ conditional quantile, leaving a residual whose $\tau^{\text{th}}$ conditional quantile is zero. This is a correct-specification condition: it rules out omitted variables that are correlated with the included regressors at the quantile of interest and requires that $\bm{x}_{it}$ be at least predetermined with respect to the quantile residual $v_{it}(\tau)$. In the empirical applications below, where the covariates include lagged GDP growth, inflation, and financial conditions, this predeterminedness is plausible, though it should be assessed in each context.

Second, part~(a)---the invariance of $f_t$ across quantiles---is the most substantive restriction in Assumption~\ref{ass:quantile_identification}. It implies that common aggregate shocks shift the entire conditional distribution of outcomes in parallel, rather than changing its shape across quantiles. This restriction rules out quantile-dependent common time shocks. Under the specification, all time-varying common variation is captured by the single factor $f_t$, while quantile-specific heterogeneity is attributed to the unit-specific intercepts and slopes. If common shocks were believed to affect the tails and the centre of the distribution differently, one could allow $f_t(\tau)$ to vary with $\tau$, but this would substantially increase the number of time-specific parameters and weaken the separation between common and unit-specific dynamics. We maintain the $\tau$-invariant specification for parsimony and because, in the applications we consider, the quantile-specific heterogeneity of primary interest operates at the unit level.

Third, although parts~(a)--(d) are stated as classical identification conditions, the model is estimated in a Bayesian framework with a hierarchical Gaussian-process prior on the unit-specific parameters $\bm{\theta}_i(\tau)$. When the time dimension is short, unit-specific coefficients may be weakly identified from the likelihood alone. The hierarchical prior provides regularisation and partial pooling across units and quantiles, thereby stabilising posterior inference. This regularisation should be distinguished from identification from the sampling model: Assumption~\ref{ass:quantile_identification} describes the restrictions under which the conditional quantile parameters are well defined, whereas the prior controls finite-sample uncertainty and shrinkage.

In addition to these identifying restrictions, we impose a standard regularity condition on the conditional distribution of the outcome. This condition is not a separate source of identification; rather, it guarantees that the conditional quantile is locally well defined and that the associated quantile regression objective has stable local behaviour around the identified parameter values.
\begin{assumption}
The conditional distribution of $y_{it}$ given $\mathcal I_{it}$ is continuous, and its conditional density is strictly positive and finite at the conditional quantile
    \[
q_{it}(\tau) = a_i(\tau)+f_t+\bm{x}_{it}'\bm{\beta}_i(\tau).
    \]
That is, for some constants $0<c<C<\infty$,
    \[
c \le f_{y\mid \mathcal I} \left(q_{it}(\tau)\mid \mathcal I_{it}\right) \le C .
    \]
\end{assumption}

Assumption~2 is the standard regularity condition for quantile regression (see, e.g., \citet{koenker2005quantile}, Theorem~4.1, and \citet{angrist2006quantile}). Continuity of the conditional distribution rules out mass points in the outcome, which is plausible for the continuously measured macroeconomic variables we study (GDP growth, consumption growth, investment growth). The requirement that the conditional density is bounded away from zero and infinity, $0<c\le f_{y\mid\mathcal I}(q_{it}(\tau)\mid\mathcal I_{it})\le C<\infty$, ensures that the conditional quantile is locally unique and that the Hessian of the check-loss objective is well behaved. In the Bayesian framework adopted below, this condition also guarantees that the asymmetric Laplace working likelihood provides a valid characterisation of the quantile of interest \citep{sriram2013posterior}.

The uniform lower bound $c>0$ deserves brief comment. In the extreme tails (say, $\tau<0.05$ or $\tau>0.95$), the conditional density of a continuously distributed outcome naturally thins, and the assumption of a {\em uniform} positive bound across all $(i,t)$ pairs may be strained. In practice, the Bayesian estimator based on the asymmetric Laplace likelihood is robust to mild violations of this condition, and the finite-sample regularisation provided by the hierarchical prior further stabilises inference in the tails. We therefore view Assumption~2 as an asymptotic idealisation that justifies the use of quantile regression methods, while recognising that its literal satisfaction at extreme quantiles is an approximation.

To address the high dimensionality of allowing unit-specific coefficients at each quantile, one can impose a random-coefficients structure. Specifically, define
\[
\bm{\theta}_i(\tau) =
\begin{pmatrix}
\alpha_i(\tau) \\
\bm{\beta}_i(\tau)
\end{pmatrix},
\]
which collects the unit-specific intercept and slope parameters at quantile $\tau$. We assume that
\begin{equation}
\label{eq:theta_prior}
\bm{\theta}_i(\tau) \sim \mathcal{N}\left(\bm{\mu}(\tau), \bm{\Sigma}(\tau)\right),
\end{equation}
independently across $i$, where $\bm{\mu}(\tau)$ is the mean vector and $\bm{\Sigma}(\tau)$ the covariance matrix. This hierarchical specification provides two advantages. First, it introduces partial pooling: the estimates of $\bm{\theta}_i(\tau)$ borrow strength across units, stabilizing inference when $T$ is small. Second, the covariance structure $\bm{\Sigma}(\tau)$ allows for systematic heterogeneity in how different units respond to the covariates at different quantiles, while keeping the model parsimonious and estimable.

Rather than treating $\bm{\theta}_i(\tau)$ as independent across quantiles, it is natural to view the collection $\{\bm{\theta}_i(\tau): \tau \in (0,1)\}$ as a stochastic process indexed by $\tau$. In particular, we may assume that $\bm{\theta}_i(\tau)$ follows a multivariate Gaussian process with mean function $\bm{\mu}(\tau)$ and covariance kernel $K(\tau,\tau')$, i.e.
\[
\bm{\theta}_i(\cdot) \sim \mathcal{GP}\big(\bm{\mu}(\cdot), K(\cdot,\cdot)\big).
\]
This formulation links parameters across quantiles smoothly, so that estimates at nearby quantiles borrow strength from one another. Intuitively, the Gaussian process prior regularizes the quantile-specific heterogeneity by ruling out arbitrary jumps across $\tau$, while still allowing flexible nonlinear variation in both intercepts and slopes. As a result, the model can capture how the entire conditional distribution of outcomes shifts across units, rather than treating quantiles in isolation.

Specifically, we write
\[
\bm{\theta}_i(\tau) = \bm{\mu}(\tau) + \bm{u}_i(\tau),
\]
where the mean function is represented as
\[
\bm{\mu}(\tau) = \bm{\Gamma}\bm{\Phi}(\tau),
\]
with $\bm{\Phi}(\tau)$ denoting a vector of basis functions (e.g.\ polynomials, splines, or Fourier terms) and $\bm{\Gamma}$ the associated coefficients. This flexible specification captures smooth systematic trends across quantiles. The mean function $\bm{\mu}(\tau) = \bm{\Gamma}\bm{\Phi}(\tau)$ captures the central trend of coefficient evolution across quantiles. In theory, for continuous distributions, the true quantile process $\bm{\theta}(\tau)$ is a smooth function of $\tau$. We approximate this unknown smooth function using a sieve, where the basis $\bm{\Phi}(\tau)$ provides a flexible yet parsimonious representation.

\cite{li2025spline} propose \emph{spline quantile regression} (SQR), a joint estimation framework that \emph{smooths across quantile levels} by representing the regression coefficients as smoothing-spline functions of $\tau$. Their univaraite method augments the standard quantile loss with a roughness penalty in $\tau$, producing coefficient paths that vary smoothly over quantiles and are computed via a large but structured linear program. In contrast, our panel model introduces a \emph{$\tau$-invariant dynamic common factor} to capture aggregate shocks, places a \emph{multivariate Gaussian-process prior on the parameter paths} to induce probabilistic smoothness and pooling across quantiles and units. A similar formulation via GP is \cite{santos2020noncrossing}, that is placed on the \emph{quantile function across $\tau$}, i.e.\ on the \emph{output} $Q_\tau(\cdot)$ (or on a directional/linear predictor implied by it), rather than directly on the \emph{parameter paths}. Econometrically, our model as oppose to \cite{santos2020noncrossing} captures how the \emph{marginal effect of a regressor} changes continuously along the conditional distribution, while preserving a low-dimensional functional form that is easily interpretable. Statistically, the multivariate GP prior acts as a flexible regularizer that shrinks the entire coefficient trajectory toward a smooth mean function, yielding efficiency gains when the number of quantile levels is large or when individual panels are short.

For the covariance structure, we adopt a kernel that decays with the distance between quantiles, so that parameters at nearby quantiles are more strongly correlated than those far apart. A convenient choice is the exponential kernel
\[
K(\tau,\tau') = \bm{\Sigma} \exp\left(-\frac{|\tau-\tau'|}{\lambda}\right),
\]
where the diagonal elements of $\bm{\Sigma}$ control the marginal variances and $\lambda>0$ is a length-scale parameter governing how quickly dependence fades as $|\tau-\tau'|$ increases. This ensures that coefficients evolve smoothly across quantiles while allowing sufficient flexibility for different behaviour at the tails. Combined with a basis-function representation of the mean, this structure balances flexibility with parsimony in describing how heterogeneous slopes and intercepts vary across the conditional distribution.

The Gaussian process specification can be interpreted along three dimensions:
\begin{enumerate}
\item the exponential kernel $K(\tau,\tau') = \bm{\Sigma}
\exp(-|\tau-\tau'|/\lambda)$ governs dependence across quantiles such that\[  \begin{bmatrix}
    \bm{\theta}_i(\tau)\\
    \bm{\theta}_i(\tau')
\end{bmatrix} \sim \mathcal{N}\left( \begin{bmatrix}
    \bm{\Gamma} \bm{\Phi}(\tau)\\
    \bm{\Gamma}\bm{\Phi}(\tau')
\end{bmatrix}, \begin{bmatrix}
    1&& \exp\left(-\frac{|\tau-\tau'|}{\lambda}\right)\\
    \exp\left(-\frac{|\tau-\tau'|}{\lambda}\right)&& 1
\end{bmatrix}\otimes \bm{\Sigma}\right)\]
\item  the matrix $\bm{\Sigma}$ encodes the dependence across coefficients within a given quantile; for example, the intercept and slope are correlated if $\bm{\Sigma}$ has nonzero off-diagonal elements, while a diagonal $\bm{\Sigma}$ implies independence
\item conditional on the Gaussian process prior, the unit-level coefficients $\bm{\theta}_i(\tau)$ are independent draws across $i$. In the special case $\bm{\Sigma}=0 \times \bm I$, all units share the same coefficient function given by the mean $\bm{\mu}(\tau)$, so the Gaussian process collapses to a common set of coefficients with no cross-unit heterogeneity.
\end{enumerate}

\subsection{Monotonicity across quantiles and covariate scaling}
Because conditional quantiles must be nondecreasing in $\tau$, it is desirable to rule out quantile crossing by construction. Throughout this subsection we work with \emph{transformed} covariates $\bm{z}_{it}=\big(T_1(x_{it,1}),\ldots,T_K(x_{it,K})\big)'\in[0,1]^K$, where each $T_k:\mathbb{R}\to[0,1]$ is a monotone scaling map such as min--max scaling, an empirical rank transform, or a logistic squash. These are applied to the original regressor $x_{it,k}$ of \autoref{eq: panel}. The transformed regressors, labeled $\bm{z}_{it}$, are the arguments of the quantile function below. Notice that they are distinct from the observed aggregate predictors $\bm{w}_t$ that enter the time effect $f_t$.
\begin{lemma}
\label{lem:noncrossing_sufficient}
Let the conditional quantile function be
\[
Q_\tau(y_{it}\mid \mathcal I_{it}) = a_i(\tau)+f_t+\bm{z}_{it}'\bm{\beta}_i(\tau),
\]
where $\bm{z}_{it}\in[0,1]^K$ and the common time effect $f_t$ is invariant across quantiles. Suppose that for every $0<\tau<\tau'<1$,
\[
a_i(\tau')\ge a_i(\tau)
\]
and
\[
\beta_{ik}(\tau')\ge \beta_{ik}(\tau), \qquad k=1,\ldots,K.
\]
Then, for every $\bm{z}_{it}\in[0,1]^K$,
\[
Q_{\tau'}(y_{it}\mid \mathcal I_{it}) \ge Q_{\tau}(y_{it}\mid \mathcal I_{it}).
\]
Hence the fitted conditional quantile function is nondecreasing in $\tau$.
\end{lemma}
The transformation $T_k$ is introduced to place the covariates on a bounded nonnegative support. It is \emph{nonnegativity} of $\bm{z}_{it}$ that makes the componentwise monotonicity condition sufficient for global noncrossing (Lemma~\ref{lem:noncrossing_sufficient}); the upper bound $z_{it,k}\le 1$ is not needed for the Lemma itself, but is used to control the constants in the probabilistic bound of Theorem~\ref{thm:soft_noncrossing} below. When $T_k$ is strictly increasing, it preserves the ordering of the original covariate. However, nonlinear transformations such as rank transformations or logistic squashing change the scale of the regressor, so the associated coefficients should be interpreted as effects of the transformed covariates rather than effects of the original covariates.

Note that by restricting the covariate space to the $K$-dimensional unit box $[0,1]^K$, our approach is related to that of \cite{bondell2010noncrossing}. Yet, their quantile curve is fitted as a function of $x$ (e.g., via penalized splines), and non-crossing is ensured by a simple constrained optimization that ties together the fits at multiple $\tau$; there is no explicit basis or stochastic process placed on the coefficient paths as functions of $\tau$. In contrast, our specification smooths \emph{across $\tau$ at the parameter level} by treating $\bm{\theta}_i(\tau)$ as a multivariate Gaussian process in $\tau$, with (later) non-crossing enforced via a Bernstein-polynomial representation. This separates regularization over quantile levels from the choice of flexibility in $x$ and provides a direct handle on the evolution of coefficient functions across the distribution.

To implement these monotonicity requirements in a stochastic specification, let $\bm{\theta}_i(\tau) = \big(a_i(\tau),\bm{\beta}_i(\tau)'\big)' = \bm{\mu}(\tau) + \bm{u}_i(\tau)$, where $\bm{\mu}(\tau)$ is a population mean path and $\bm{u}_i(\cdot)$ is a zero-mean stochastic deviation. Panel-wise monotonicity means that for all $0<\tau<\tau'<1$, $\bm{\theta}_i(\tau')-\bm{\theta}_i(\tau)\succeq 0$ componentwise. On a grid $0<\tau_1<\cdots<\tau_L<1$, the requirement $\bm{\theta}_i(\tau_{\ell+1})-\bm{\theta}_i(\tau_\ell)\succeq 0$ (componentwise) yields \emph{$p(L-1)$ linear inequalities per unit}, i.e.\ $N\,p(L-1)$ constraints in total. Instead, we enforce $\bm{\mu}(\cdot)$ to be monotone and control the probability that the unit-specific deviation $\bm{u}_i(\cdot)$ overturns this ordering through the scale of its Gaussian-process covariance. A convenient way to parametrise a monotone mean path is to expand each component in the Bernstein basis on $[0,1]$:
\[
\mu_j(\tau)=\sum_{m=0}^{M}\theta_{j,m} b_{m,M}(\tau), \qquad
b_{m,M}(\tau)=\binom{M}{m}\tau^{m}(1-\tau)^{M-m}.
\]
The Bernstein basis $ \{b_{m,M}\}_{m=0}^M $ on $[0,1]$ is nonnegative and forms a partition of unity, $\sum_{m=0}^M b_{m,M}(\tau)=1$. Hence $\mu_j(\tau)$ always lies in the convex hull of its coefficients $\{\theta_{j,m}\}$, which makes shape constraints transparent. In particular, a sufficient and convenient condition for $\mu_j(\cdot)$ to be nondecreasing on $[0,1]$ is $\theta_{j,0}\le\cdots\le\theta_{j,M}$. Imposing this reduces to linear inequalities on $\theta_{j,\cdot}$, so estimation under monotonicity remains a convex quadratic/linear program. Bernstein polynomials also enjoy uniform approximation: any continuous (and, in particular, any monotone) function on $[0,1]$ can be approximated arbitrarily well as $M\to\infty$, while avoiding knot placement issues that arise with splines. The coefficients are interpretable: $\theta_{j,0}$ and $\theta_{j,M}$ anchor the boundary levels (near $\tau=0,1$), and intermediate $\theta_{j,m}$ approximate $\mu_j(m/M)$, providing a clear link between parameter values and the quantile index.

Rather than enforcing panel-wise monotonicity for each unit (which entails $Np(L{-}1)$ inequalities on a grid), we impose monotonicity only on the population path $\bm{\mu}(\cdot)$ and let unit-specific deviations be small and smooth in $\tau$. Choosing small diagonal elements of $\bm{\Sigma}$ (with a prior centered near zero) shrinks $\bm{\theta}_i(\cdot)$ toward the monotone mean $\bm{\mu}(\cdot)$, delivering \emph{soft} monotonicity---unit-level curves are nearly nondecreasing with high probability---while still allowing data-driven departures when warranted; the length scale $\lambda$ controls how smoothly such departures evolve across quantiles. If rare local crossings remain, a post-estimation monotone rearrangement can be applied as a diagnostic or post-processing step to restore monotonicity of the reported quantile curves, although this may slightly alter the original empirical objective.

Componentwise monotonicity of the population path $\bm{\mu}(\cdot)$ guarantees that the population conditional quantile function is nondecreasing in $\tau$. However, unit-specific GP deviations can still generate local crossings. The probability of such crossings is small only when the population quantile curve has a positive separation margin between $\tau$ and $\tau'$ relative to the variance of the GP increment. The following theorem formalizes this soft noncrossing property.

\begin{assumption}[Smooth GP Prior]
\label{ass:gp}
The unit-specific deviation $\bm{u}_i(\cdot) = (u_{i,1}(\cdot), \ldots, u_{i,p}(\cdot))'$ follows a mean-zero Gaussian process prior with separable covariance structure,
\[
\operatorname{Cov}\big(\bm{u}_i(\tau), \bm{u}_i(\tau')\big) = K_\lambda(\tau,\tau')\bm{\Sigma},
\]
where $\bm{\Sigma}\in\mathbb{R}^{p\times p}$ is a positive semidefinite matrix capturing within-quantile dependence across the $p$ intercept and slope components, and $\bm{K}_\lambda$ is a stationary kernel, normalized so that $K_\lambda(\tau,\tau)=1$ for all $\tau\in[0,1]$, and satisfying $K_\lambda(\tau,\tau')\le 1$.
\end{assumption}
Fix $0<\tau<\tau'<1$ and let
\[
\bm{r}(\bm{z})=(1,\bm{z}')'\in\mathbb{R}^{p}, \qquad \bm{z}\in\mathcal Z\subseteq[0,1]^K, \qquad p=K+1 .
\]
Write
\[
\bm{\theta}_i(\tau)=\bm{\mu}(\tau)+\bm{u}_i(\tau),
\]
where $\bm{\theta}_i(\tau)=(a_i(\tau),\bm{\beta}_i(\tau)')'$ and $\bm{u}_i(\cdot)$ satisfies Assumption~\ref{ass:gp}. Define the population margin
\[
m_{\tau,\tau'} = \inf_{\bm{z}\in\mathcal Z} \bm{r}(\bm{z})' \left\{ \bm{\mu}(\tau')-\bm{\mu}(\tau) \right\}.
\]

\begin{theorem}
\label{thm:soft_noncrossing}
Suppose that
\[
m_{\tau,\tau'}>0.
\]
Let
\[
\bm{d}_i(\tau,\tau') = \bm{u}_i(\tau')-\bm{u}_i(\tau).
\]
Under Assumption~\ref{ass:gp}, each component $d_{ij}(\tau,\tau')$ is mean-zero Gaussian with variance $\operatorname{Var}(d_{ij}(\tau,\tau')) = 2\{1-K_\lambda(\tau,\tau')\}\Sigma_{jj}$. Define the maximum marginal variance
\[
\bar\nu^2_{\tau,\tau'} = 2\{1-K_\lambda(\tau,\tau')\}\max_{1\le j\le p} \Sigma_{jj}.
\]
Then, for a given unit $i$, {writing $Q_\tau^{\,i}(\bm{z})=a_i(\tau)+f_t+\bm{z}'\bm{\beta}_i(\tau)$ for the fitted conditional quantile at a generic covariate value $\bm{z}\in\mathcal Z$,}
\[
\Pr\left( \exists \bm{z}\in\mathcal Z: {Q_{\tau'}^{\,i}(\bm{z})} <
{Q_{\tau}^{\,i}(\bm{z})} \right) \le p \exp\left( - \frac{m_{\tau,\tau'}^2}
{2p^2\bar\nu^2_{\tau,\tau'}} \right).
\]
In particular, for the exponential kernel
\[
K_\lambda(\tau,\tau') = \exp\left(-\frac{|\tau-\tau'|}{\lambda}\right),
\]
the bound becomes
\[
\Pr\left( \exists \bm{z}\in\mathcal Z: {Q_{\tau'}^{\,i}(\bm{z})} <
{Q_{\tau}^{\,i}(\bm{z})} \right) \le p \exp\left[ - \frac{m_{\tau,\tau'}^2}
{4p^2 \max_{1\le j\le p}\Sigma_{jj} \left\{1-\exp(-|\tau'-\tau|/\lambda)\right\}}
\right].
\]
Thus, for fixed $\tau<\tau'$ and fixed positive margin $m_{\tau,\tau'}>0$, the probability of crossing converges to zero as $\max_j\Sigma_{jj}\to 0$.
\end{theorem}

Theorem~\ref{thm:soft_noncrossing} relies on two modelling choices that merit comment. First, the theorem uses only the maximum marginal variance $\max_j\Sigma_{jj}$ rather than requiring independent components; the union bound remains valid under arbitrary within-quantile dependence because it bounds the probability that {\em any} component of the GP deviation is large enough to cause crossing. If $\bm{\Sigma}$ is diagonal, the bound reduces to the independent-components case, while nonzero off-diagonal entries (which imply correlated intercept and slope deviations) leave the bound unchanged. The bound is therefore consistent with the general separable covariance structure $\bm{K}_\lambda\otimes\bm{\Sigma}$ used throughout the paper.

Second, the noncrossing guarantees---both the deterministic {Lemma~\ref{lem:noncrossing_sufficient}} and the probabilistic Theorem~\ref{thm:soft_noncrossing}---assume that the transformed covariates lie in $[0,1]^K$ with all components nonnegative. This is satisfied by construction after applying the scaling transformations {$T_k$ introduced above} (min--max scaling, rank transforms, or logistic squashing). Without such a transformation, {Lemma~\ref{lem:noncrossing_sufficient}} would not hold for covariates that can take negative values, and the noncrossing property would need to be enforced through alternative means, such as constrained optimisation over the full covariate support. The transformation thus provides a simple sufficient condition for global noncrossing at the cost of reinterpreting the coefficients as effects of the transformed, rather than original, covariates.

\subsection{Bayesian Estimation}

Quantile regression can be conveniently estimated in a Bayesian framework by exploiting the connection between the check loss function and the asymmetric Laplace distribution (ALD). Specifically, the $\tau^{\text{th}}$ conditional quantile of $y_{it}$ given covariates $\bm{x}_{it}$ can be represented as the location parameter of an ALD with density
\[
f_{\tau}(u \mid \sigma) = \frac{\tau(1-\tau)}{\sigma} \exp\left\{ -
\rho_{\tau}\left(\frac{u}{\sigma}\right) \right\},
\]
where $u = y_{it} -a_i(\tau)-f_t -\bm{x}_{it}'\bm{\beta}_i(\tau)$, $\sigma>0$ is a scale parameter, and $\rho_{\tau}(u) = u(\tau - \mathbf{1}\{u<0\})$ denotes the check function. This distributional representation provides a likelihood-based interpretation of quantile regression, enabling straightforward Bayesian inference through standard simulation methods. An important property of the asymmetric Laplace distribution is that it admits a location–scale mixture representation involving a normal and an exponential distribution (see \cite{kozumi2011gibbs}).

Let $\mathcal{T}=\{\tau_1,\ldots,\tau_L\}$ denote the quantile grid and write the quantile residual as
\[
u_{it}(\tau) = y_{it}-a_i(\tau)-f_t-\bm{x}_{it}'\bm{\beta}_i(\tau).
\]
For each $\tau\in\mathcal{T}$ we introduce latent mixing variables $\Omega_\tau=\{\omega_{it,\tau}\}_{i\le N,t\le T}$.
Using the ALD mixture representation, for each $\tau$ define
\[
\kappa_\tau=\frac{1-2\tau}{\tau(1-\tau)}, \qquad \phi_\tau=\frac{2}{\tau(1-\tau)},
\]
and allow either a common scale $\sigma$ or quantile-specific scales $\{\sigma_\tau\}$. We adopt the unit-mean parametrisation of the mixing variable, $\omega_{it,\tau}\sim\operatorname{Exp}(1)$, under which the asymmetric Laplace residual admits the location--scale mixture
\[
u_{it}(\tau) = \kappa_\tau\sigma_\tau\omega_{it,\tau} +
\sigma_\tau\sqrt{\phi_\tau\omega_{it,\tau}}\;e_{it,\tau}, \qquad
e_{it,\tau}\sim\mathcal{N}(0,1),
\]
so that, conditionally on $\omega_{it,\tau}$, the residual is Gaussian with mean $\kappa_\tau\sigma_\tau\omega_{it,\tau}$ and variance $\sigma_\tau^2\phi_\tau\omega_{it,\tau}$. Conditionally on $(\bm{\Theta}_{\mathcal{T}},f_{1:T})$ with $\bm{\Theta}_{\mathcal{T}}=\{\bm{\theta}_i(\tau)=(a_i(\tau),\bm{\beta}_i(\tau)')'\}_{i\le N,\tau\in\mathcal{T}}$, the complete-data likelihood of $(y,\Omega_{\mathcal{T}})$ across all quantiles factorizes as
\begin{multline*}
\mathcal{L}^{\star}(y,\Omega_{\mathcal{T}}\mid
\bm{\Theta}_{\mathcal{T}},f_{1:T},\{\sigma_\tau\})
= \prod_{\tau\in\mathcal{T}} \prod_{i=1}^N\prod_{t=1}^T
\phi\Big( y_{it}; a_i(\tau)+f_t+\bm{x}_{it}'\bm{\beta}_i(\tau)+\kappa_\tau\sigma_\tau\omega_{it,\tau}, \\
\sigma_\tau^{2}\phi_\tau\omega_{it,\tau} \Big)\exp(-\omega_{it,\tau})
\end{multline*}
subject to the normalization $\sum_{t=1}^T f_t=0$. Here the latent term $\exp(-\omega_{it,\tau})$ is the $\operatorname{Exp}(1)$ density of $\omega_{it,\tau}$ and carries no dependence on $\sigma_\tau$; the only $\sigma_\tau$-dependence outside the exponent is the Gaussian normalising factor $(2\pi\sigma_\tau^2\phi_\tau\omega_{it,\tau})^{-1/2}$ retained inside $\phi(\cdot;\cdot,\cdot)$. Because $\sigma_\tau$ is itself sampled, this factor (which contributes $\sigma_\tau^{-1}$ per observation) must be kept; only constants free of all model parameters are omitted from the complete-data likelihood.

With proper initialization $f_1\sim\mathcal{N}(m_0,C_0)$, the state likelihood is
\[
\mathcal{L}_{\text{state}}(\gamma,\bm{\xi},\sigma_f^2 \mid f_{1:T},\bm{w}_{1:T}) = \phi\big(f_1;
m_0,C_0\big) \prod_{t=2}^T \phi\big(f_t; \gamma f_{t-1}+\bm{\xi}'\bm{w}_t, \sigma_f^2\big).
\]
The conditional posterior distribution of the dynamic factors $f_{1:T}$ is Gaussian on a hyperplane, and their sampling follows the standard \emph{precision sampler}~\citep{chan2023high} used in linear Gaussian state-space models. Equivalently, one may introduce an unconstrained latent process $g_t$ satisfying the state equation $g_t=\gamma g_{t-1}+\bm{\xi}'\bm{w}_t+\eta_t$ and define $f_t=g_t-\bar g$, where $\bar g=T^{-1}\sum_{t=1}^T g_t$. This centering imposes $\sum_{t=1}^T f_t=0$ without changing the interpretation of $f_t$ as a common time effect. Given Gaussian and inverse-Gamma priors, the posterior of the state equation parameters can be sampled in very standard Gibbs steps.

Stack quantile parameters over quantiles:
\[
\bm{\Theta}_i = \big(\bm{\theta}_i(\tau_1)',\ldots,\bm{\theta}_i(\tau_L)'\big)' \in \mathbb{R}^{pL},
\qquad \bm{\mu} = \big(\bm{\mu}(\tau_1)',\ldots,\bm{\mu}(\tau_L)'\big)' .
\]
Place a separable GP prior across quantiles and coefficients,
\[
\bm{\Theta}_i \sim \mathcal{N}\big(\bm{\mu}, \bm{K}_\lambda \otimes \bm{\Sigma}\big), \quad
K_\lambda(\tau,\tau')=\exp\big(-|\tau-\tau'|/\lambda\big)
\]
independently across $i$, where $\bm{K}_\lambda\in\mathbb{R}^{L\times L}$ correlates nearby quantiles and $\bm{\Sigma}\in\mathbb{R}^{p\times p}$ couples intercept/slope components.

To make the update explicit, collect the within-unit design and working response. For unit $i$ let $\bm{r}_{it}=(1,\bm{x}_{it}')'\in\mathbb{R}^p$ and $\bm{R}_i=(\bm{r}_{i1},\ldots,\bm{r}_{iT})'\in\mathbb{R}^{T\times p}$, and for each quantile $\tau_\ell$ define the working response and working precision
\[
\tilde y_{it,\ell}=y_{it}-f_t-\kappa_{\tau_\ell}\sigma_{\tau_\ell}\omega_{it,\ell}, \qquad
\bm{W}_{i\ell}=\operatorname{diag}\left[\frac{1}{\sigma_{\tau_\ell}^2\phi_{\tau_\ell}\omega_{it,\ell}}\right]_{t=1}^T,
\]
so that, conditional on the latent scales, the mixture representation gives $\tilde{\bm{y}}_{i\ell}=\bm{R}_i\bm{\theta}_i(\tau_\ell)+\bm{\varepsilon}_{i\ell}$ with $\bm{\varepsilon}_{i\ell}\sim\mathcal{N}(0,\bm{W}_{i\ell}^{-1})$. Because the likelihood factorizes across quantiles while the GP prior couples them, the conditional posterior of $\bm{\Theta}_i$ is Gaussian,
\[
\bm{\Theta}_i\mid\cdot\sim\mathcal{N}(\bar{\bm{m}}_i,\bar{\bm{V}}_i),
\]
with
\[
\bar{\bm{V}}_i^{-1} = (\bm{K}_\lambda^{-1}\otimes\bm{\Sigma}^{-1}) +
\operatorname*{blockdiag}_{\ell=1}^L\big(\bm{R}_i'\bm{W}_{i\ell}\bm{R}_i\big), \qquad \bar{\bm{m}}_i = \bar{\bm{V}}_i\Big[ (\bm{K}_\lambda^{-1}\otimes\bm{\Sigma}^{-1})\bm{\mu} +
\operatorname*{vec}_{\ell}\big\{\bm{R}_i'\bm{W}_{i\ell}\tilde{\bm{y}}_{i\ell}\big\} \Big].
\]
The prior precision $\bm{K}_\lambda^{-1}\otimes\bm{\Sigma}^{-1}$ borrows strength across adjacent quantiles, while the block-diagonal data precision $\bm{R}_i'\bm{W}_{i\ell}\bm{R}_i$ enters quantile by quantile; both are of modest dimension ($pL$), so the draw is a single multivariate Gaussian sample.

Stack all Bernstein coefficients into
\[
\bm{\Gamma} =
\begin{bmatrix}
\bm{\vartheta}_1'\\[-2pt]
\vdots\\[-2pt]
\bm{\vartheta}_p'
\end{bmatrix}
\in\mathbb{R}^{p\times(M+1)}.
\]
Let
\[
\bm{\gamma} = \operatorname{vec}(\bm{\Gamma}') \in \mathbb{R}^{p(M+1)}.
\]
The mean path $\bm{\mu}(\tau)=\bm{\Gamma}\bm{\Phi}(\tau)$ is monotone in $\tau$ when each row of $\bm{\Gamma}$ is nondecreasing, i.e.
\[
\bm{D}\bm{\Gamma}' \succeq 0, \qquad \bm{D} =
\begin{bmatrix}
-1 & 1 & 0 & \cdots & 0\\
0 & -1 & 1 & \cdots & 0\\
\vdots & & \ddots & \ddots & \vdots\\
0 & \cdots & 0 & -1 & 1
\end{bmatrix}
\in\mathbb{R}^{M\times(M+1)}.
\]
The row-wise monotonicity constraints are equivalent to
\[
(\bm{I}_p \otimes \bm{D})\bm{\gamma} \succeq 0.
\]
We therefore assign a joint truncated Gaussian prior
\[
\bm{\gamma} \sim \mathcal{N}_{\mathcal{C}}\big(\bm{m}_\Gamma,\bm{V}_\Gamma\big), \qquad
\mathcal{C}=\{\bm{\gamma}\in\mathbb{R}^{p(M+1)}:(\bm{I}_p\otimes \bm{D})\bm{\gamma}\succeq 0\},
\]
with density (up to normalization)
\[
\pi(\bm{\Gamma}) \propto \phi\big(\bm{\gamma};\bm{m}_\Gamma,\bm{V}_\Gamma\big)
\mathbf{1}\{(\bm{I}_p\otimes \bm{D})\bm{\gamma}\succeq 0\},
\]
where $\bm{m}_\Gamma$ and $\bm{V}_\Gamma$ are the mean vector and covariance of a multivariate Gaussian prior.

The conditional posterior of $\bm{\gamma}$ follows from the same conjugate structure. Writing the population path as a linear map of $\bm{\gamma}$, $\bm{\mu}(\tau_\ell)=(\bm{I}_p\otimes\bm{\Phi}(\tau_\ell)')\bm{\gamma}$, stack these maps over the quantile grid into the design
\[
\bm{\mu}=\bm{\Psi}\bm{\gamma}, \qquad \bm{\Psi}=
\big[\bm{I}_p\otimes\bm{\Phi}(\tau_1)',\ldots,\bm{I}_p\otimes\bm{\Phi}(\tau_L)'\big]'
\in\mathbb{R}^{pL\times p(M+1)} .
\]
Since $\bm{\Theta}_i\mid\bm{\gamma}\sim\mathcal{N}(\bm{\Psi}\bm{\gamma},\bm{K}_\lambda\otimes\bm{\Sigma})$ independently across $i=1,\ldots,N$, combining the units with the truncated Gaussian prior yields
\[
\bm{\gamma}\mid\cdot\sim \mathcal{N}_{\mathcal C}\big(\bar{\bm{m}}_\Gamma,\bar{\bm{V}}_\Gamma\big),
\qquad \mathcal C=\{\bm{\gamma}:(\bm{I}_p\otimes \bm{D})\bm{\gamma}\succeq 0\},
\]
where
\[
\bar{\bm{V}}_\Gamma^{-1} = \bm{V}_\Gamma^{-1} + N\bm{\Psi}'(\bm{K}_\lambda^{-1}\otimes\bm{\Sigma}^{-1})\bm{\Psi}, \qquad
\bar{\bm{m}}_\Gamma = \bar{\bm{V}}_\Gamma\Big[ \bm{V}_\Gamma^{-1}\bm{m}_\Gamma +
\bm{\Psi}'(\bm{K}_\lambda^{-1}\otimes\bm{\Sigma}^{-1})\sum_{i=1}^N\bm{\Theta}_i \Big].
\]
This is the unconstrained Gaussian update $\mathcal{N}(\bar{\bm{m}}_\Gamma,\bar{\bm{V}}_\Gamma)$ truncated to the monotone cone $\mathcal C$.

Sampling from the joint truncated Gaussian can be challenging, as the linear inequality constraints define a high-dimensional convex polytope with correlated components. Fortunately, the dimensionality of $\bm{\Gamma}$ in our implementation is modest, since the Bernstein basis representation already reduces the parameter space to a small set of smooth coefficients. This makes it feasible to employ the efficient algorithm of \citet{botev2017normal}, which provides exact and numerically stable simulation from multivariate truncated normal distributions using minimax tilting. As a result, posterior sampling of $\bm{\Gamma}$ remains computationally tractable even under joint monotonicity constraints.

\section{Climate Shocks and International Downside Risk}
\label{sec:climate}

Standard empirical work on the relationship between climate shocks and the macroeconomy estimates the average response of macroeconomic quantities such as output growth or inflation to temperature shocks. The implicit assumption is that climate shocks shift the centre of the distribution of a macroeconomic target without altering its shape \citep{colacito2019temperature, 10.1093/qje/qjag011}.
This assumption is questionable on both theoretical and empirical grounds. Real-options theory predicts that uncertainty about future conditions suppresses irreversible investment \citep{bloom2018really}, depressing the lower tail of growth realisations without necessarily moving the median \citep{jovanovic2022uncertainty, huang2024financial}. Similarly, international transmission channels may amplify climate tail risks heterogeneously, with commodity importers bearing the bulk of downside risk from global climate disturbances while exporters of affected commodities enjoy countercyclical terms-of-trade gains. These considerations motivate studying how climate shocks impact the conditional distribution of, e.g., output growth rather than their average effects alone.

In this section, we apply the Bayesian quantile panel model of Section \ref{sec: econometric_framework} and use it to analyze how temperature shocks and temperature volatility shocks shape the entire conditional distribution of output growth across 33 countries over 1979Q3--2023Q3. To do so, we cast the quantile panel model in a local projections framework and distinguish between \emph{level} shocks---unexpected deviations in temperature---and \emph{volatility} shocks---unexpected increases in the variance of temperature fluctuations---and within each type between \emph{local} (idiosyncratic) and \emph{global} (common) components. This two-by-two decomposition allows us to test three nested hypotheses: (i) temperature level shocks generate asymmetric downside risks relative to the mean; (ii) temperature volatility shocks compound this asymmetry, shifting the left tail independently of the median; and (iii) the global component of climate shocks constitutes a systemic, non-diversifiable risk that is absent from country-specific models.

\subsection{International macroeconomic data and the measurement of temperature shocks}
\label{subsec:climate_data}
We construct a balanced quarterly panel of 33 countries spanning the period 1979Q3--2023Q3.\footnote{The original raw data from \citet{mohaddes2024compilation} start in 1979Q2. However, since GDP is transformed by taking log differences, the first usable observation is 1979Q3.} Macroeconomic data are sourced from the 2023 vintage of the GVAR database compiled by \citet{mohaddes2024compilation}, which draws primarily on Haver Analytics, the International Monetary Fund's International Financial Statistics (IFS), and Bloomberg. The 33 countries are organised into five regional groups---(1) Asia-Pacific, (2) North America, (3) South America, (4) the Middle East and Africa, and (5) Europe---and together account for more than 90 per cent of world GDP \citep{mohaddes2024compilation}.\footnote{The countries in each group, together with the abbreviations used in the remainder of the paper, are as follows. \emph{Asia-Pacific}: Australia (AUS), China (CHN), India (IND), Indonesia (IDN), Japan (JPN), Korea (KOR), Malaysia (MYS), New Zealand (NZL), the Philippines (PHL), Singapore (SGP), and Thailand (THA). \emph{North America}: Canada (CAN), Mexico (MEX), and the United States (US). \emph{South America}: Argentina (ARG), Brazil (BRA), Chile (CHL), and Peru (PER). \emph{Middle East and Africa}: Saudi Arabia (SAU) and South Africa (ZAF). \emph{Europe}: Austria (AUT), Belgium (BEL), Finland (FIN), France (FRA), Germany (DEU), Italy (ITA), the Netherlands (NLD), Norway (NOR), Spain (ESP), Sweden (SWE), Switzerland (CHE), Turkey (TUR), and the United Kingdom (UK).} A full list of countries is provided in Table~\ref{tab:countries} in the Data Appendix.

From this database, we extract five macroeconomic variables: real GDP, inflation, the short-term interest rate, the real exchange rate, and a global oil price index. As is standard in the local projections literature, real GDP is transformed into a cumulative annualised growth rate at each horizon $h$, defined as
\begin{equation}
    y_{i,t+h} = \frac{400}{h} \sum_{j=1}^{h} \Delta \log(\text{GDP}_{i,t+j}),
    \label{eq:gdp}
\end{equation}
where the factor of 400 converts log-differences of quarterly data to annualised percentage points.

Climate data are obtained from the unified country-level temperature dataset of \citet{gortan2024unified},\footnote{Available at https://weightedclimatedata.streamlit.app.} which provides monthly average surface temperatures for each country in our sample. Since the raw temperature series exhibit strong seasonal patterns, we remove seasonality prior to constructing our shock and volatility measures. Following \citet{colacito2019temperature}, we regress each country's temperature series on a full set of monthly dummy variables and retain the residuals as the deseasonalised temperature series. These residuals capture deviations from the estimated seasonal cycle and are free of predictable calendar variation. The deseasonalised monthly series are then aggregated to quarterly frequency to align with the macroeconomic data.

Temperature shocks are constructed following the autoregressive filtering approach of \citet{10.1093/qje/qjag011}. For each country~$i$, we estimate the regression
\begin{equation}
    \widehat{T}_{it}
    = T_{it} - \Bigl(\hat{\alpha}_i
    + \hat{\beta}_{i,1}\,T_{i,t-h}
    + \cdots
    + \hat{\beta}_{i,P+1}\,T_{i,t-h-P}\Bigr),
    \label{eq:ts}
\end{equation}
where $T_{it}$ denotes the deseasonalised temperature for country~$i$ in period~$t$, $\hat{\beta}_{i,j}$ is the coefficient estimate on the lag of order $h + j - 1$, and $\widehat{T}_{it}$ is the resulting temperature shock, defined as the residual deviation from the estimated autoregressive component. In their original annual-frequency analysis, \citet{10.1093/qje/qjag011} set $P = 2$ and $h = 2$. Since our empirical analysis is conducted at the quarterly frequency, we set $P = 8$ and $h = 8$, which preserves the same span of dependence in calendar time. Each country's AR coefficients are estimated by OLS over the full sample. The resulting residuals $\widehat{T}_{it}$ measure quasi-exogenous temperature surprises after removing predictable inertia in the local temperature process.

Figure~\ref{fig:tempshock} displays the estimated temperature shock series for a selection of advanced and emerging market economies. The series exhibit considerable cross-country heterogeneity, both in the magnitude and persistence of shocks, consistent with the view that climate exposures differ markedly across regions and development stages.

\begin{figure}[htbp]
    \centering
    \includegraphics[width=\textwidth]{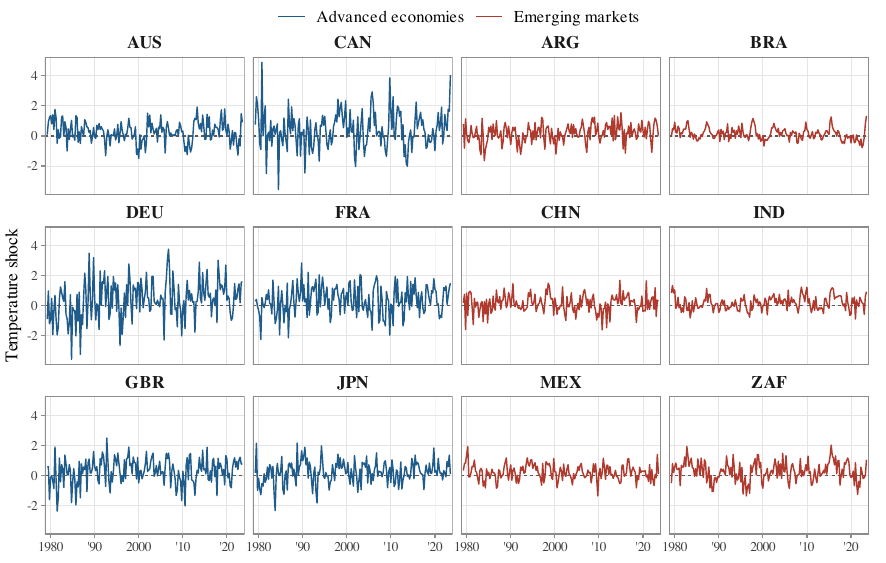}
    \caption*{\scriptsize \textit{Notes}: The Advanced Economies are Australia, Canada, Germany, France, UK and Japan. The Emerging Economies are Argentina, Brazil, China, India, Mexico and South Africa.}
    \caption{Temperature shocks for selected advanced and emerging market economies.}
    \label{fig:tempshock}
\end{figure}
Temperature volatility is measured as the realised volatility of the temperature shock series, constructed using a rolling window estimator. Specifically, for each country, we compute
\begin{equation}
    \widehat{\sigma}_{it}^{\text{RV}} = \sqrt{\frac{1}{4}
    \sum_{j=0}^{3} \left(\widehat{T}_{i,t-j}\right)^{2}},
    \label{eq:rv}
\end{equation}
where $\widehat{T}_{i,t-j}$ denotes the temperature shock defined in equation~\eqref{eq:ts}, and $\widehat{\sigma}_{it}^{\text{RV}}$ is the root mean square (RMS) of shocks over the preceding four quarters. This estimator is analogous to a time-varying standard deviation under the assumption of zero mean shocks, and captures within-year fluctuations in climate variability over a one-year rolling window.

Figure~\ref{fig:tempvol} displays the estimated temperature volatility series for the same selection of advanced and emerging market economies shown in Figure~\ref{fig:tempshock}. As with the shock series, substantial cross-country heterogeneity is evident, though the ordering broadly tracks latitude rather than income: volatility is highest in the mid- and high-latitude economies of the sample, such as Germany and Canada, and lowest in the tropical ones, such as Brazil and India.

\begin{figure}[H]
    \centering
    \includegraphics[width=\textwidth]{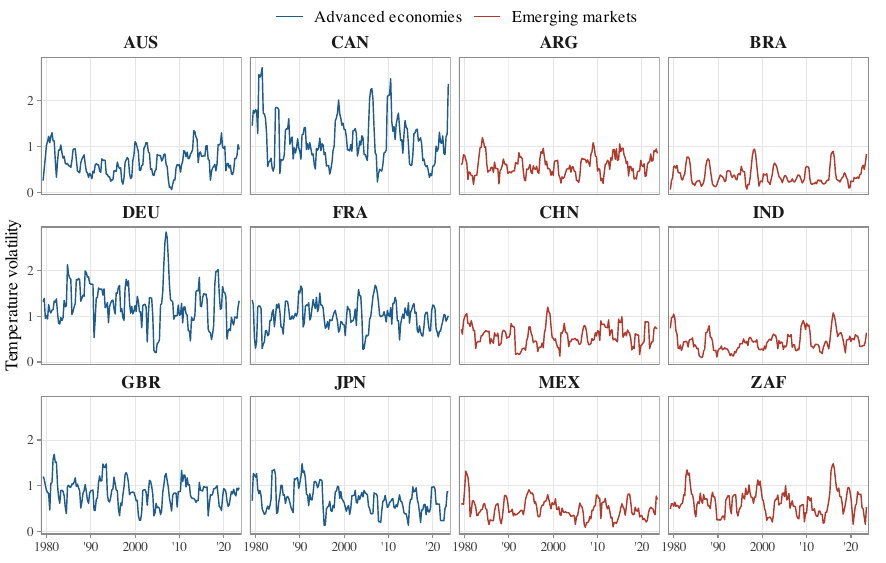}
    \caption*{\scriptsize \textit{Notes}: The Advanced Economies are Australia, Canada, Germany, France, UK and Japan. The Emerging Economies are Argentina, Brazil, China, India, Mexico and South Africa. Temperature volatility is computed as the root mean square of temperature shocks over a rolling window of four quarters, as defined in equation~\eqref{eq:rv}.}
    \caption{Temperature volatility for selected advanced and emerging market economies.}
    \label{fig:tempvol}
\end{figure}

To distinguish domestic climate conditions from common global climate disturbances, we decompose the temperature shock series defined in equation~\eqref{eq:ts} into local and global components. The global temperature shock is constructed as the GDP-weighted cross-sectional average
\begin{equation}
    \widehat{T}^{G}_{t} = \sum_{i=1}^{N} \omega_i
    \widehat{T}_{it},
    \qquad \sum_{i=1}^{N} \omega_i = 1,
    \label{eq:globalshock}
\end{equation}
where $\widehat{T}_{it}$ is the temperature shock for country $i$ in period $t$, and $\omega_i$ denotes country $i$'s GDP weight. The local temperature shock is then defined as the deviation of each country's shock from the global component,
\begin{equation}
    \widehat{T}^{L}_{it} = \widehat{T}_{it} -
    \widehat{T}^{G}_{t}.
    \label{eq:localshock}
\end{equation}
Each country's shock therefore splits additively into the two components, $\widehat{T}_{it} = \widehat{T}^{L}_{it} + \widehat{T}^{G}_{t}$, with the local pieces satisfying $\sum_{i=1}^{N} \omega_i \widehat{T}^{L}_{it} = 0$ by construction. The undecomposed shock $\widehat{T}_{it}$ thus sits one level above the local--global split, which is why it carries no superscript.

An analogous decomposition is applied to the realised volatility measure defined in equation~\eqref{eq:rv},
\begin{equation}
    \widehat{\sigma}^{G,\text{RV}}_{t} = \sum_{i=1}^{N} \omega_i
    \widehat{\sigma}_{it}^{\text{RV}}, \qquad
    \widehat{\sigma}^{L,\text{RV}}_{it} = \widehat{\sigma}_{it}^{\text{RV}}
    - \widehat{\sigma}^{G,\text{RV}}_{t}.
    \label{eq:localvol}
\end{equation}
Following \citet{mohaddes2024compilation}, GDP weights are computed using each country's average GDP at purchasing power parity (PPP) over the period 2014--2016, ensuring that the global factor reflects the relative economic size of each country in our panel.

In what follows, we let each climate disturbance carry a pair of superscripts. The first records whether the shock is local ($L$) or global ($G$) while the second indicates whether it is a temperature level shock ($T$) or a temperature volatility shock ($V$). The four regressors defined above are therefore $\widehat{T}^{L}_{it}$ and $\widehat{T}^{G}_{t}$ for the level shocks, and $\widehat{\sigma}^{L,\text{RV}}_{it}$ and $\widehat{\sigma}^{G,\text{RV}}_{t}$ for the volatility shocks. We index the resulting shock types by $s \in \mathcal{S} = \{(L,T), (G,T), (L,V), (G,V)\}$.

\begin{figure}[h!]
    \centering
    \includegraphics[width=\textwidth]{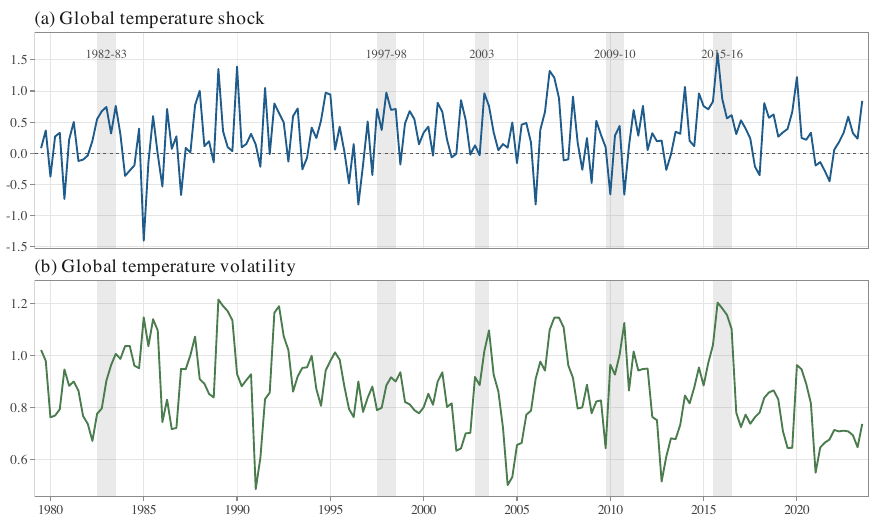}
    \caption*{\scriptsize \textit{Notes}: The global temperature shock and volatility series are constructed as GDP-PPP-weighted cross-sectional averages of the country-level measures defined in equations~\eqref{eq:ts} and~\eqref{eq:rv}, respectively, following equations~\eqref{eq:globalshock} and~\eqref{eq:localvol}. Shaded bands denote major El Ni\~{n}o episodes: 1982--83, 1997--98, 2002--03, 2009--10, and 2015--16.}
    \caption{Global temperature shocks and volatility, 1979Q3--2023Q3.}
    \label{fig:tempglobal}
\end{figure}

Figure~\ref{fig:tempglobal} plots the constructed global temperature shock and volatility series over the sample period 1979Q3--2023Q3. Five shaded bands indicate episodes of major El Ni\~{n}o events: 1982--83, 1997--98, 2002--03, 2009--10, and 2015--16. These episodes coincide with pronounced spikes in both the shock and volatility series.

\subsection{Model specification}
\label{subsec:climate_specification}
For each horizon $h = 1, \ldots, H$, we estimate the following panel quantile local projection:
\begin{equation}
\begin{aligned}
    Q_{\tau}(y_{i,t+h} \mid \mathcal{F}_t)
    =  & \alpha_{i,h}(\tau) + f_{t,h}  + \widehat{T}^{L}_{it}\theta^{L,T}_{i,h}(\tau)
      + \widehat{T}^{G}_{t}\theta^{G,T}_{i,h}(\tau)  \\
    & + \widehat{\sigma}^{L,\text{RV}}_{it}\theta^{L,V}_{i,h}(\tau)
      + \widehat{\sigma}^{G,\text{RV}}_{t}\theta^{G,V}_{i,h}(\tau)
      + \bm{x}_{it}'\bm{\beta}_{i,h}(\tau),
\end{aligned}
\label{eq:qLP}
\end{equation}
where $y_{i,t+h}$ is the cumulative annualised GDP growth rate defined in equation~\eqref{eq:gdp}, $Q_{\tau}(\cdot \mid \mathcal{F}_t)$ denotes the $\tau^{\text{th}}$ conditional quantile given the information set $\mathcal{F}_t$, $\alpha_{i,h}(\tau)$ is a country-specific intercept, and $f_{t,h}$ is a horizon-specific common time effect that is \emph{invariant across quantiles}, exactly as in the model of Section~\ref{sec: econometric_framework}. Consistent with that specification, $f_{t,h}$ is modelled as a \emph{restricted} latent factor driven by the observed global predictors $\bm{w}_t$ (an AR(1) process with innovation variance $\sigma_f^2$), rather than as an unrestricted set of time dummies.

This restriction is not merely cosmetic; it is required for identification of the global climate effects. The global regressors $\widehat{T}^{G}_{t}$ and $\widehat{\sigma}^{G,\text{RV}}_{t}$ vary only over time and are common to all countries, so they are collinear with any \emph{unrestricted} time effect. Writing the unit-specific loading as a level plus a deviation, $\theta^{G,T}_{i,h}(\tau)=\bar\theta^{G,T}_{h}(\tau)+\tilde\theta^{G,T}_{i,h}(\tau)$, the common component $\widehat{T}^{G}_{t}\bar\theta^{G,T}_{h}(\tau)$ is a pure function of $t$ and would be absorbed by a full set of time dummies, leaving the \emph{average} global climate effect unidentified; only the heterogeneous deviations $\widehat{T}^{G}_{t}\tilde\theta^{G,T}_{i,h}(\tau)$ would survive. By contrast, the restricted factor $f_{t,h}$ lives in a low-dimensional space and does not span the time variation in $\widehat{T}^{G}_{t}$, so both the average and heterogeneous global effects are identified. Because $f_{t,h}$ retains a free innovation $\eta_t$ each period, this separation is achieved through the AR(1) restriction together with the prior on $\sigma_f^2$ and the heterogeneity of the loadings, rather than through the likelihood alone; we therefore keep $\sigma_f^2$ tightly centred and exclude the global climate regressors from $\bm{w}_t$. The local and global climate regressors $\widehat{T}^{L}_{it}$, $\widehat{T}^{G}_{t}$, $\widehat{\sigma}^{L,\text{RV}}_{it}$, and $\widehat{\sigma}^{G,\text{RV}}_{t}$ are constructed as described in equations~\eqref{eq:localshock}--\eqref{eq:localvol}, and $\tau \in \mathcal{T} = \{0.01, 0.02, 0.03, \ldots, 0.98, 0.99\}$ is the quantile index. We estimate equation~\eqref{eq:qLP} separately for each horizon $h = 1, \ldots, 10$ quarters.

A separate concern with the local projection specification in equation~\eqref{eq:qLP} is that constructing $y_{i,t+h}$ as a cumulative $h$-quarter growth rate mechanically induces a moving-average structure of order $h-1$ in the regression residuals, since observations at adjacent origin dates $t$ share overlapping periods of the underlying growth process. We address this concern via simulation rather than through an explicit correction: Appendix~\ref{subsec:MA_simulation}
evaluates whether posterior credible interval coverage is sensitive to this induced dependence, using a Monte Carlo design that replicates the overlapping-horizon structure while holding the true quantile coefficients fixed across $h$. Coverage remains close to the nominal 90 per cent level at every horizon and coefficient (0.89--0.91 across $R=100$ replications), with no systematic decline as $h$ increases, providing no evidence that the induced MA($h-1$) structure distorts inference in our setting.

The baseline control vector is
\begin{equation}
    \bm{x}_{it} = \bigl[
        y_{it},\ \pi_{it},\ \Delta rer_{it}
    \bigr]', \qquad
    \bm{w}_t = \bigl[\text{oilp}_{t},\ r^{*}_{t}\bigr]',
    \label{eq:controls}
\end{equation}
where $y_{it}$ denotes lagged output growth, $\pi_{it}$ is inflation, $\Delta rer_{it}$ is the change in the real exchange rate, $\text{oilp}_{t}$ is the global oil price, and $r^{*}_{t}$ is the global short-term interest rate. The vector $\bm{w}_t$ collects the observed aggregate predictors that drive the common factor $f_{t,h}$, as in Section~\ref{sec: econometric_framework}; it does not enter equation~\eqref{eq:qLP} directly. All variables are sourced from the GVAR database of \citet{mohaddes2024compilation} as described in the previous section.

The coefficients of interest are $\theta^{L,T}_{i,h}(\tau)$, $\theta^{G,T}_{i,h}(\tau)$, $\theta^{L,V}_{i,h}(\tau)$, and $\theta^{G,V}_{i,h}(\tau)$, which measure the effect of local temperature shocks, global temperature shocks, local temperature volatility, and global temperature volatility, respectively, on the $\tau^{\text{th}}$ quantile of the conditional distribution of future output growth at horizon $h$. Estimating equation~\eqref{eq:qLP} across the full quantile grid $\mathcal{T}$ traces out the entire conditional distribution of $y_{i,t+h}$, allowing us to assess whether climate shocks shift the mean, compress or widen the distribution, or generate asymmetric tail risks.

Prior to estimation, all four climate regressors---$\widehat{T}^{L}_{it}$, $\widehat{T}^{G}_{t}$, $\widehat{\sigma}^{L,\text{RV}}_{it}$, and $\widehat{\sigma}^{G,\text{RV}}_{t}$---are standardised to have unit variance in the cross-time distribution. This scaling means that estimated coefficients in equation~\eqref{eq:qLP} below are interpretable as the response to a one-standard-deviation climate disturbance, facilitating direct comparison across shock types.

\subsection{Baseline Distributional Responses}
\label{subsec:climate_baseline_results}

We now present the baseline distributional impulse responses of cumulative output growth to each of the four climate shocks. For shock type $s$, horizon $h$, and quantile $\tau$, the distributional impulse response function (DIRF) is defined directly from the estimated coefficients of equation~\eqref{eq:qLP}:
\begin{equation}
    \text{DIRF}^{s}_{i,h}(\tau) = \theta^{s}_{i,h}(\tau).
    \label{eq:DIRF}
\end{equation}
At a given $\tau$, equation~\eqref{eq:DIRF} is the response of the $\tau^{\text{th}}$ conditional quantile. The distributional response is the collection of these ordinates across the full grid, $\{\theta^{s}_{i,h}(\tau)\}_{\tau \in \mathcal{T}}$, since it is the entire coefficient path that determines how the shape of the predictive distribution changes; this path is what we report in the figures below.

To summarise the cross-country evidence, we report the panel-average response
\begin{equation}
    \overline{\text{DIRF}}^{\,s}_{h}(\tau)
    = \frac{1}{N} \sum_{i=1}^{N} \text{DIRF}^{s}_{i,h}(\tau),
    \label{eq:avgDIRF}
\end{equation}
which aggregates country-level responses into a single summary measure at each quantile $\tau \in \mathcal{T}$ and horizon $h$. We report posterior medians together with 68 and 90 per cent credible bands throughout. The central hypothesis of this paper is that temperature volatility shocks generate disproportionately large contractions in the lower tail of the output growth distribution relative to the median, that is,
\begin{equation}
    \overline{\text{DIRF}}^{\,s}_{h}(0.10)
    < \overline{\text{DIRF}}^{\,s}_{h}(0.50), \qquad s \in \{(L,V), (G,V)\},
    \label{eq:hypothesis}
\end{equation}
particularly at medium-run horizons ($h = 4$--$10$ quarters). Evidence in favour of equation~\eqref{eq:hypothesis} would indicate that temperature volatility acts primarily as a tail risk, compressing downside outcomes without necessarily shifting the median forecast, consistent with the uncertainty channel identified in the climate-macroeconomics literature.

\begin{figure}[h!]
    \centering
    \includegraphics[width=\textwidth]{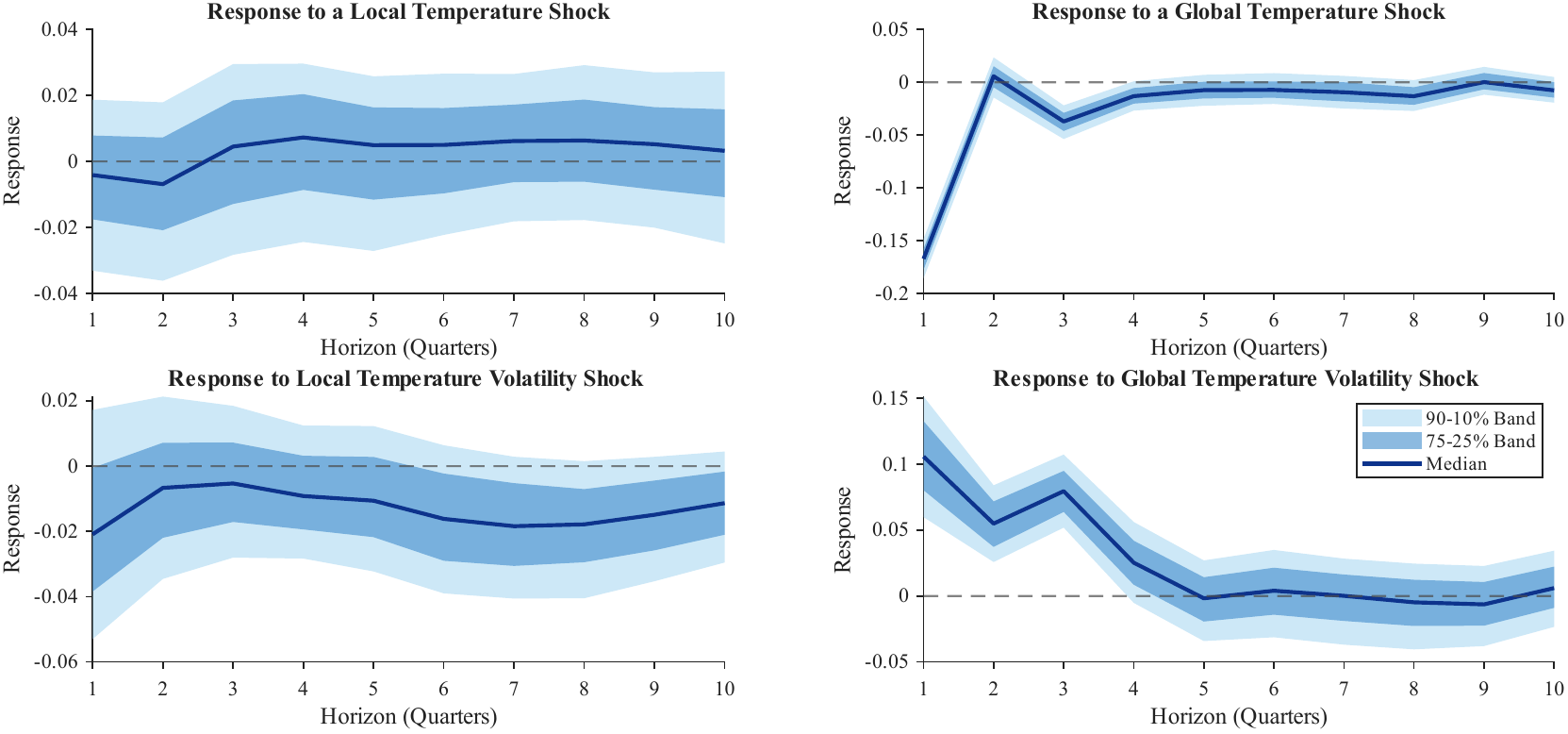}
    \caption*{\scriptsize \textit{Notes}: Each panel plots $\overline{\text{DIRF}}^{\,s}_{h}(\tau)$ as defined in equation~\eqref{eq:avgDIRF} for $h = 1, \ldots, 10$ quarters. The solid line is the median response ($\tau = 0.50$). Dark and light shaded bands correspond to the 25$^{\text{th}}$--75$^{\text{th}}$ and 10$^{\text{th}}$--90$^{\text{th}}$ quantile responses, respectively. Panels report responses to: (a) local temperature shock; (b) global temperature shock; (c) local temperature volatility shock; (d) global temperature volatility shock.}
    \caption{Panel-average distributional impulse responses to local and global climate shocks.}
    \label{fig:AverageIRFs}
\end{figure}

Figure~\ref{fig:AverageIRFs} plots $\overline{\text{DIRF}}^{\,s}_{h}(\tau)$ across forecast horizons $h = 1,\ldots,10$ quarters for each of the four climate shocks. Several findings emerge. First, local temperature shocks generate modest and broadly symmetric responses across the conditional distribution, with the effects centred around zero and the 90--10 per cent quantile range spanning approximately $[-0.04,0.04]$ percentage points. This finding is consistent with the existing literature, which documents substantial heterogeneity in the economic effects of temperature shocks across countries. In particular, the adverse consequences of temperature increases appear to be concentrated in countries with hotter climates and lower levels of income, where economic activity is more vulnerable to climatic conditions (see \citet{sachs1997sources, jones2010climate, dell2012temperature, burke2015global}).

In contrast, global temperature shocks generate more pronounced negative effects, particularly in the lower tail of the conditional distribution. The response at $\tau=0.10$ declines to approximately $-0.2$ percentage points at short-to-medium horizons, while responses at higher quantiles remain close to zero or mildly positive. This pattern indicates the presence of asymmetric downside risks associated with the global climate disturbances, whereby adverse climate developments disproportionately affect weaker economic outcomes. This result is consistent with the findings of \citet{kahn2021long}, who show that persistent increases in temperatures above historical norms are associated with significant reductions in long-run economic growth.

Third, local temperature volatility shocks produce a modest response that is skewed toward the lower tail, with the 90--10 per cent band spanning approximately $[-0.06,0.02]$ percentage points and widening somewhat at longer horizons. This asymmetry suggests that increases in local climate uncertainty are associated with a mild downward shift in the lower tail of the output growth distribution, rather than a purely symmetric widening of uncertainty around an unchanged centre. Finally, global temperature volatility shocks display a different pattern, generating a positive response in output growth in the short term that dissipates over the medium to long term.

One possible explanation for this unexpected positive response is that global temperature volatility may partly capture broader fluctuations associated with global commodity markets. For example, \citet{brunner2002nino} show that El Niño--Southern Oscillation events can raise world commodity prices and affect global real GDP through international transmission channels. A positive global temperature volatility shock may therefore generate heterogeneous effects across countries through commodity-price and terms-of-trade channels. Countries that benefit from higher commodity prices may experience stronger output responses, which could explain the positive effects observed in the upper quantiles of the distribution.

Turning to the central hypothesis in equation~\eqref{eq:hypothesis}, the evidence presents a mixed picture at the panel level. While lower quantile responses lie below the median for local volatility shocks at most horizons, the global volatility shock generates the opposite pattern, with upper quantiles diverging sharply from the median. This tension between the two volatility components suggests that the aggregate panel average may conceal systematic differences across economies at different stages of development. We investigate this possibility below by examining country-level responses separately for advanced and emerging market economies.

\subsection{Cross-Country Heterogeneity}
\label{subsec:heterogeneity}

The panel-average responses in Figure~\ref{fig:AverageIRFs} provide a useful summary of the aggregate evidence but may obscure considerable heterogeneity across economies. To examine this, we present country-level distributional impulse responses for a selection of ten advanced economies (AE) --- Australia, Canada, France, Germany, Italy, Japan, Norway, Spain, the United Kingdom, and the United States --- and ten emerging market economies (EM) --- Argentina, Brazil, China, India, Indonesia, Korea, Mexico, South Africa, Thailand, and Turkey. Figures~\ref{fig:IRFlocaltemp}--\ref{fig:IRFglobalvol} plot $\text{DIRF}^{s}_{i,h}(\tau)$ as defined in equation~\eqref{eq:DIRF} for each country, shock, and quantile across horizons $h = 1, \ldots, 10$ quarters.

\paragraph{Local temperature shocks.}
Figure~\ref{fig:IRFlocaltemp} shows that responses to local temperature shocks are generally modest for advanced economies, with 90--10 per cent quantile bands typically spanning approximately $[-0.02, 0.04]$ percentage points; Germany, Japan, and Norway exhibit somewhat wider bands, of approximately $[-0.05, 0.05]$. This is broadly consistent with idiosyncratic domestic temperature fluctuations having limited systematic macroeconomic impact in developed economies. Emerging markets display substantially greater cross-country dispersion: India exhibits markedly wider quantile bands than any advanced economy, spanning approximately $[-0.10, 0.10]$, while Thailand and Turkey show comparably wide bands. China is notable for a band that lies entirely in non-negative territory. These results indicate that local temperature shocks generate meaningful, but heterogeneous, downside risks across several emerging economies, consistent with \citet{berg2024gdp}, who similarly document considerable heterogeneity in the sign and magnitude of GDP responses to idiosyncratic temperature shocks across countries, with the direction of the response varying systematically with country characteristics such as income, educational attainment, and trade openness.

\paragraph{Global temperature shocks.}
Figure~\ref{fig:IRFglobaltemp} shows that global temperature shocks generate a negative response across essentially all economies. Most countries -- including Canada, France, Germany, Italy, Japan, Norway, Spain, the United Kingdom, and the emerging markets in the sample -- exhibit quantile bands extending to approximately $-0.20$ percentage points. Australia and the United States are notable exceptions, with substantially narrower bands not exceeding approximately $-0.15$ and $-0.10$ respectively, suggesting a comparatively muted response. Indonesia is the only country whose band extends into positive territory. These results indicate that global temperature shocks generate a broadly shared downside risk to output growth across both advanced and emerging economies, consistent with \citet{10.1093/qje/qjag011} and \citet{byrne2024macroeconomic}, who similarly find that global climate disturbances generate negative macroeconomic effects shared broadly across both advanced and emerging economies, in contrast to the more limited and heterogeneous effects of country-specific temperature shocks.

\paragraph{Local temperature volatility shocks.}
Figure~\ref{fig:IRFlocalvol} shows that, with the exception of a handful of countries such as France, the United Kingdom, the United States, India, Indonesia, and Mexico, most economies exhibit a negative median response to a local temperature volatility shock in the short term, consistent with \citet{alessandri2025macroeconomic}, who show that a rise in temperature risk -- controlling for temperature levels -- is followed by a period of both lower and more volatile GDP growth. Among advanced economies, Canada, Italy, Japan, and Spain exhibits persistently negative median responses through much of the horizon, while Australia, France, the United Kingdom, and the United States display a mildly positive median response instead. Among emerging markets, the negative short-term response is more pronounced: Argentina, Brazil, Korea, and Thailand exhibit sharply negative median responses on impact, with Brazil falling to approximately $-0.11$ percentage points, while China and Turkey begin negative before reverting toward positive territory at longer horizons. By contrast, India, Indonesia, and Mexico display persistently positive median responses throughout. A negative short-term response to local temperature volatility is thus the dominant pattern across the panel, though the persistence and magnitude of this response varies considerably by country, and the negative lower-tail responses observed for several emerging markets provide some country-level support for the hypothesis in equation~\eqref{eq:hypothesis}, even where the panel average does not strongly confirm it.

\paragraph{Global temperature volatility shocks.}
Finally, figure~\ref{fig:IRFglobalvol} shows that a global temperature volatility shock generates a large positive response on impact for every economy in the panel, with median responses ranging from approximately $0.07$ to $0.20$ percentage points on impact, before decaying over the following four to five quarters. This impact-horizon uniformity is not fully explained by the commodity-price channel discussed above and is left as an open question for future work.

\begin{figure}[H]
    \centering
    \includegraphics[width=\textwidth]{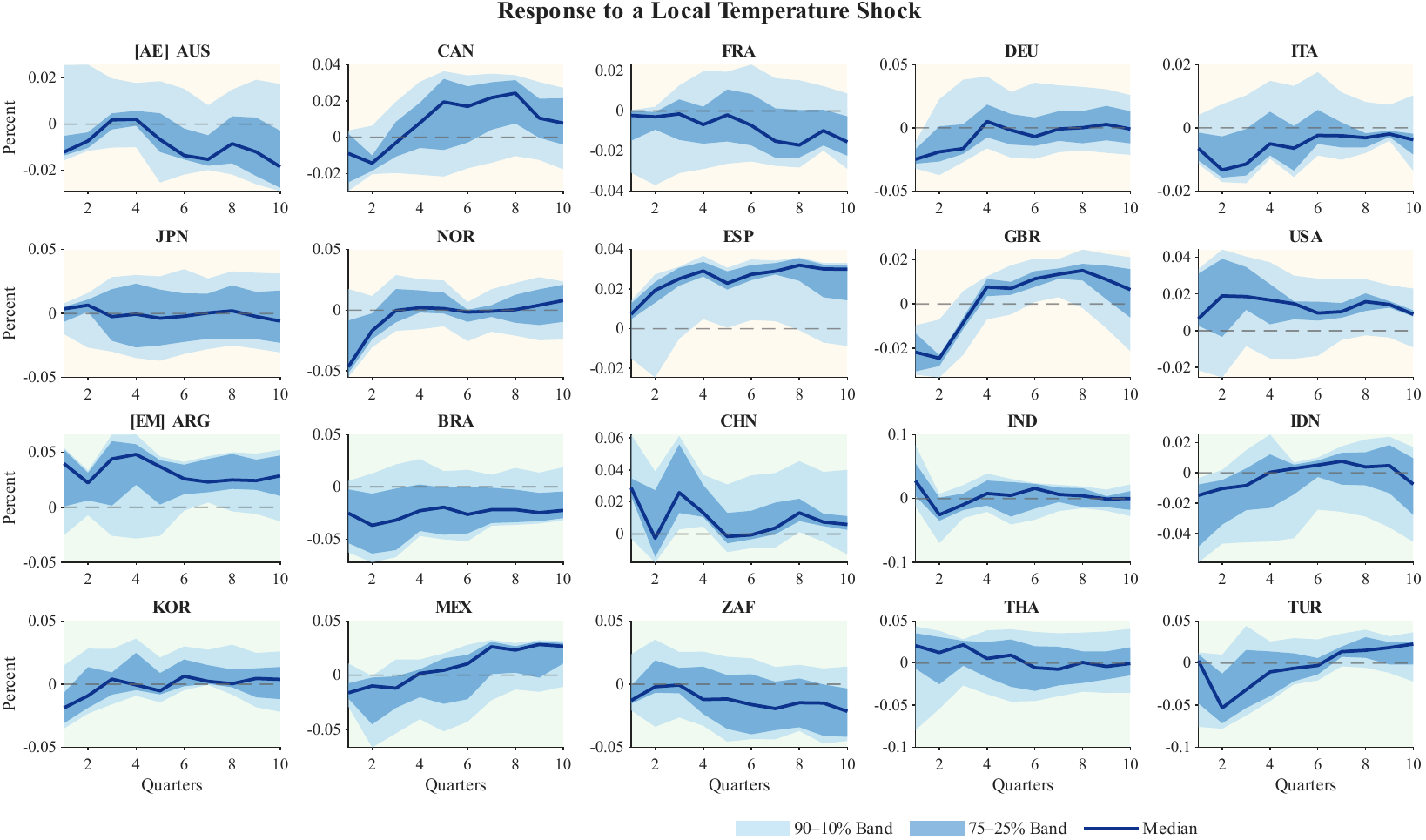}
    \caption*{\scriptsize \textit{Notes}: Each panel plots $\text{DIRF}^{(L,T)}_{i,h}(\tau)$ as defined in equation~\eqref{eq:DIRF} for $h = 1,\ldots,10$ quarters. The solid line is the median response ($\tau = 0.50$). Dark and light shaded bands correspond to the 25$^{\text{th}}$--75$^{\text{th}}$ and 10$^{\text{th}}$--90$^{\text{th}}$ quantile responses, respectively. Countries labelled \textup{[AE]} and \textup{[EM]} denote advanced and emerging market economies, respectively.}
    \caption{Country-level distributional impulse responses to a local temperature shock.}
    \label{fig:IRFlocaltemp}
\end{figure}

\begin{figure}[H]
    \centering
    \includegraphics[width=\textwidth]{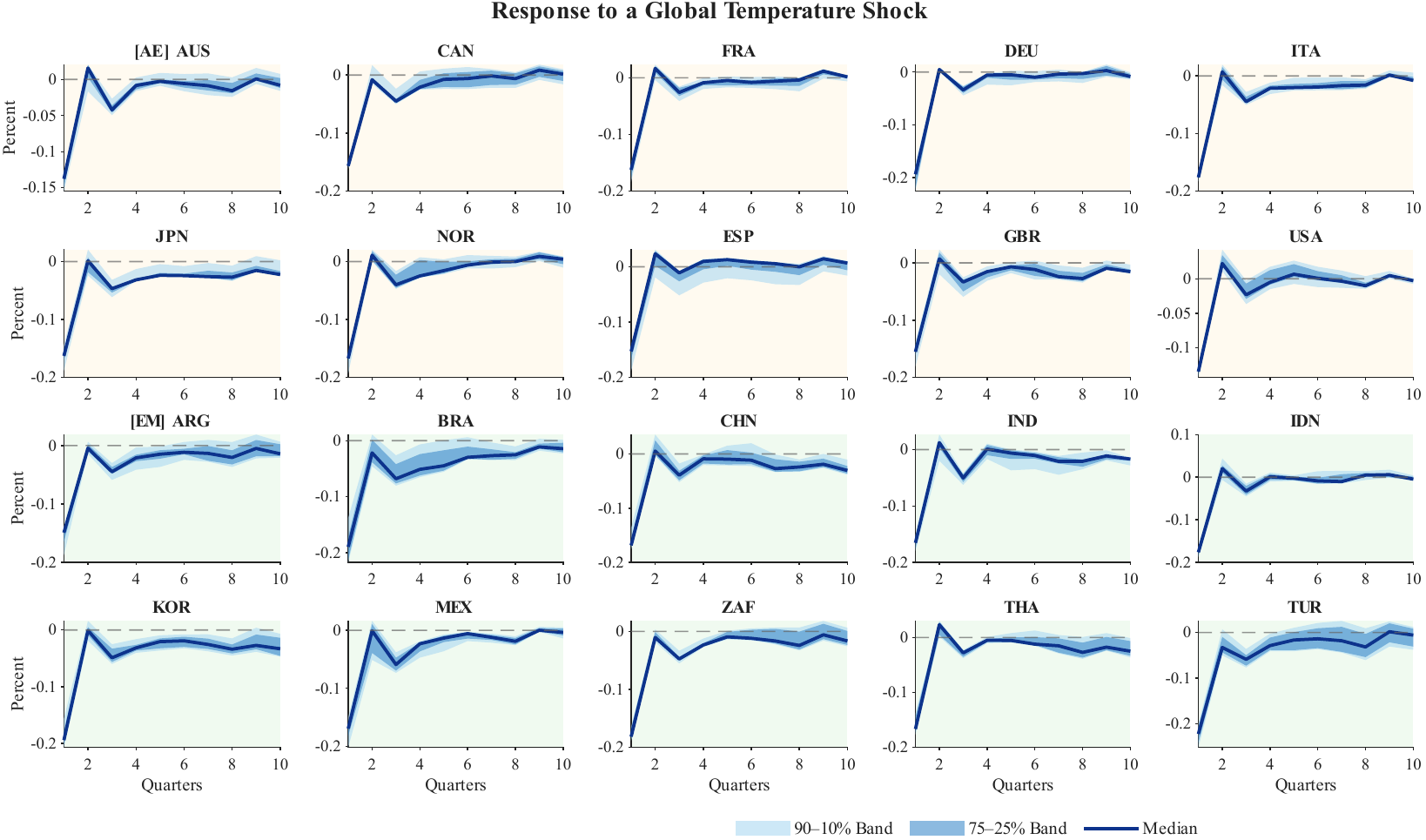}
    \caption*{\scriptsize \textit{Notes}: Each panel plots $\text{DIRF}^{(G,T)}_{i,h}(\tau)$ as defined in equation~\eqref{eq:DIRF} for $h = 1,\ldots,10$ quarters. See notes to Figure~\ref{fig:IRFlocaltemp} for further details.}
    \caption{Country-level distributional impulse responses to a global temperature shock.}
    \label{fig:IRFglobaltemp}
\end{figure}

\begin{figure}[H]
    \centering
    \includegraphics[width=\textwidth]{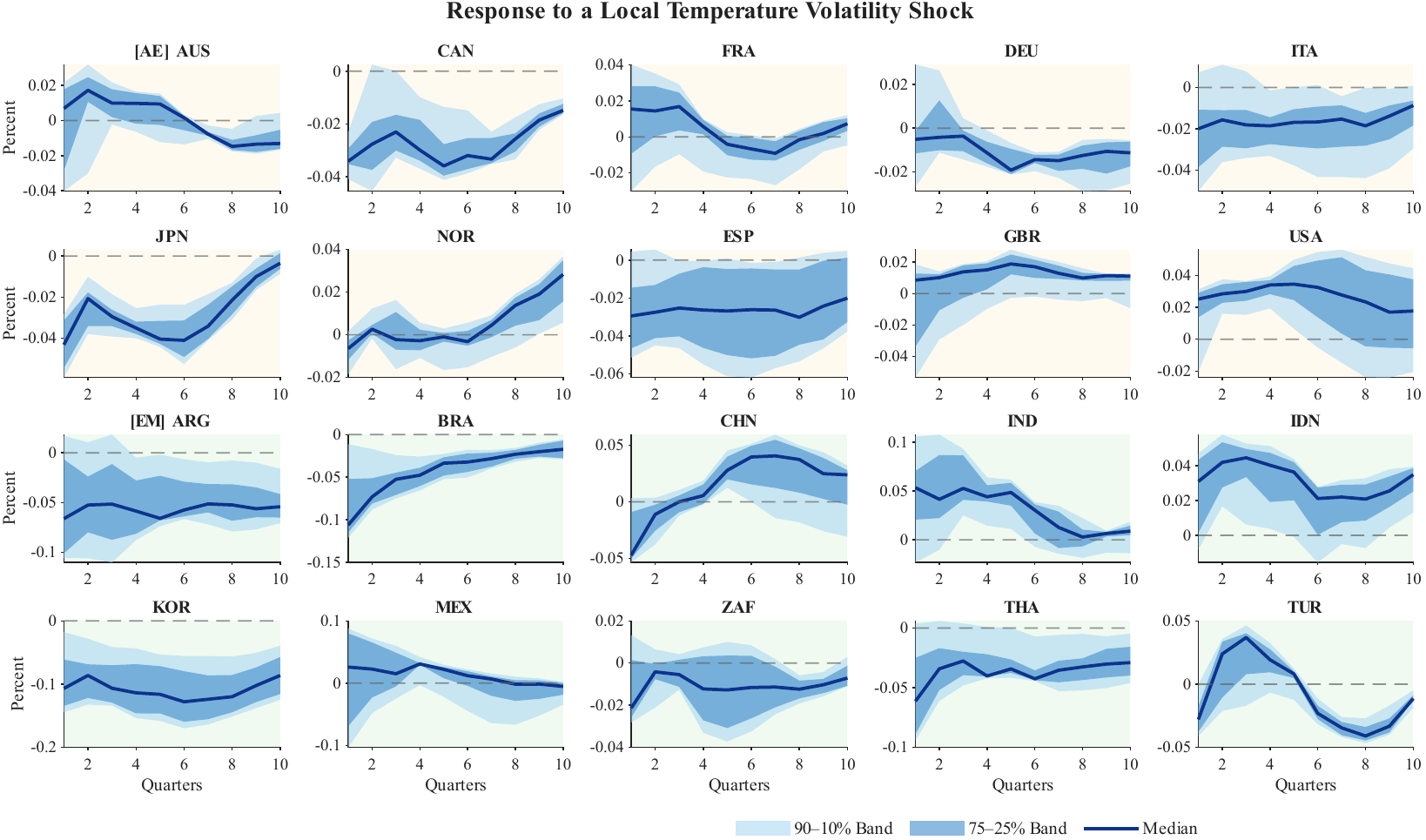}
    \caption*{\scriptsize \textit{Notes}: Each panel plots $\text{DIRF}^{(L,V)}_{i,h}(\tau)$ as defined in equation~\eqref{eq:DIRF} for $h = 1,\ldots,10$ quarters. See notes to Figure~\ref{fig:IRFlocaltemp} for further details.}
    \caption{Country-level distributional impulse responses to a local temperature volatility shock.}
    \label{fig:IRFlocalvol}
\end{figure}

\begin{figure}[H]
    \centering
    \includegraphics[width=\textwidth]{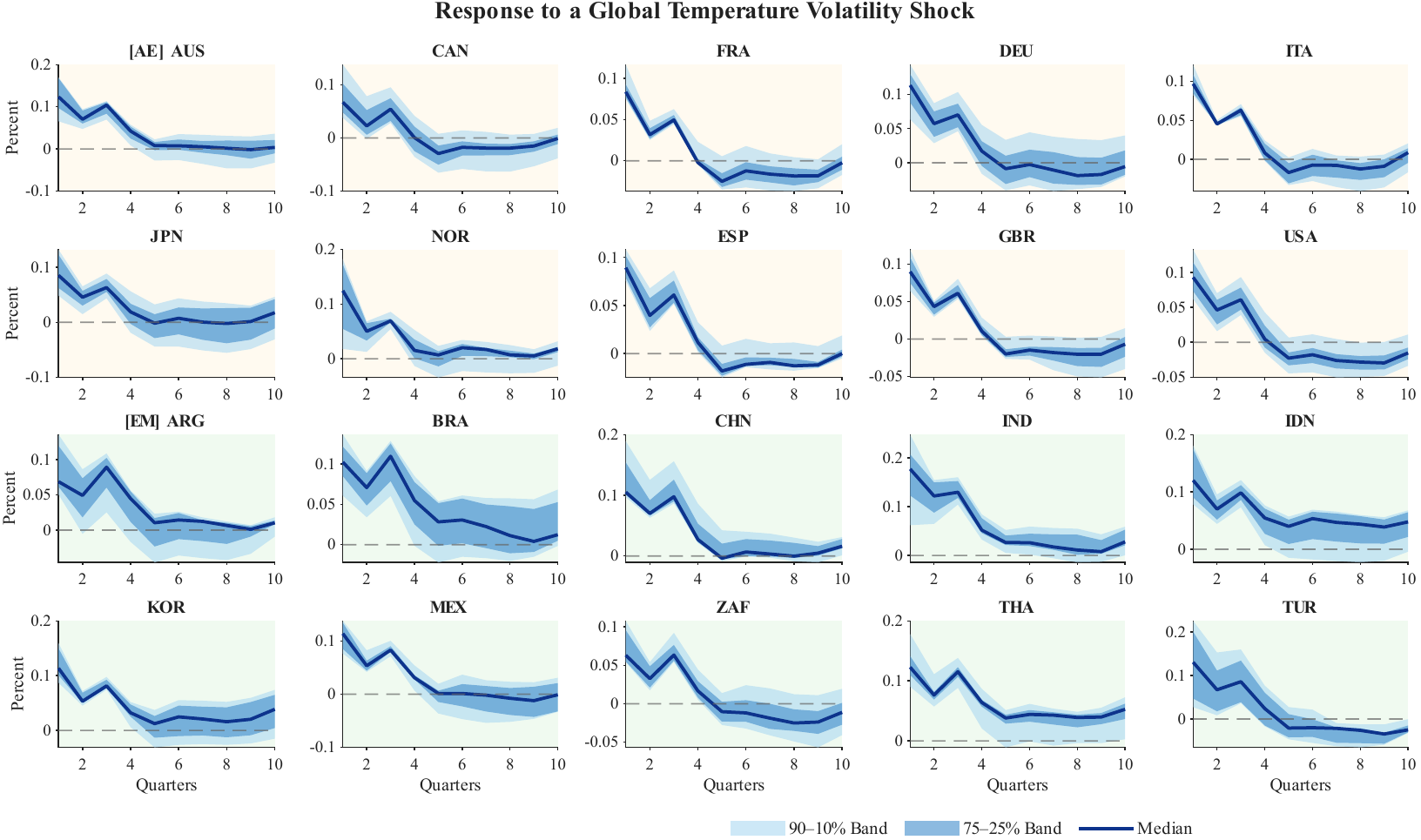}
    \caption*{\scriptsize \textit{Notes}: Each panel plots $\text{DIRF}^{(G,V)}_{i,h}(\tau)$ as defined in equation~\eqref{eq:DIRF} for $h = 1,\ldots,10$ quarters. See notes to Figure~\ref{fig:IRFlocaltemp} for further details.}
    \caption{Country-level distributional impulse responses to a global temperature volatility shock.}
    \label{fig:IRFglobalvol}
\end{figure}

\subsection{Climate Growth-at-Risk}
\label{subsec:climate_gar}

To translate the estimated quantile responses into an economically interpretable measure of downside risk, we compute climate growth-at-risk (GaR). For country $i$ and horizon $h$, define
\begin{equation}
    \text{GaR}_{i,h}^{s}(\tau) = Q_{\tau}\left(y_{i,t+h}
    \mid \mathcal{F}_t, s\right),
    \label{eq:GaR}
\end{equation}
where $s$ denotes a counterfactual climate scenario and $Q_{\tau}(\cdot \mid \mathcal{F}_t, s)$ is the $\tau^{\text{th}}$ conditional quantile of output growth under scenario $s$. The baseline scenario sets all climate regressors --- $\widehat{T}^{L}_{it}$, $\widehat{T}^{G}_{t}$, $\widehat{\sigma}^{L,\text{RV}}_{it}$, and $\widehat{\sigma}^{G,\text{RV}}_{t}$ --- to zero, corresponding to a counterfactual in which no climate disturbance occurs. The shock scenario increases a single climate regressor by one standard deviation, holding all others at zero. The climate impact on growth-at-risk is then measured by the change
\begin{equation}
    \Delta\text{GaR}_{i,h}^{s}(\tau) = \text{GaR}_{i,h}^{s}(\tau)
    - \text{GaR}_{i,h}^{0}(\tau),
    \label{eq:DeltaGaR}
\end{equation}
where $\text{GaR}_{i,h}^{0}(\tau)$ denotes the baseline. A negative value of $\Delta\text{GaR}_{i,h}^{s}(\tau)$ indicates that the climate shock shifts the $\tau^{\text{th}}$ quantile of the output growth distribution downward, representing a deterioration in downside risk. We focus primarily on the left tail of the conditional distribution, corresponding to either the 5$^{\text{th}}$ or 10$^{\text{th}}$ quantile, which are the standard thresholds in the growth-at-risk literature \citep{adrian2019vulnerable}.

To summarise the full left-tail impact of each climate scenario, we complement the quantile-specific GaR measure with the expected shortfall:
\begin{equation}
    \text{ES}_{i,h}^{s}(\alpha) = \frac{1}{\alpha}
    \int_0^{\alpha} Q_u\left(y_{i,t+h} \mid
    \mathcal{F}_t, s\right) du,
    \label{eq:ES}
\end{equation}
which averages the conditional quantile function over the entire lower tail $[0, \alpha]$ and thereby captures the severity of adverse outcomes beyond a given threshold. On the estimated quantile grid $\mathcal{T}$, equation~\eqref{eq:ES} is approximated by
\begin{equation}
    \widehat{\text{ES}}_{i,h}^{s}(\alpha) =
    \frac{1}{|\mathcal{T}_{\alpha}|}
    \sum_{\tau \in \mathcal{T}_{\alpha}}
    \widehat{Q}_{\tau}\left(y_{i,t+h} \mid
    \mathcal{F}_t, s\right),
    \qquad
    \mathcal{T}_{\alpha} = \left\{\tau \in \mathcal{T} :
    \tau \leq \alpha \right\},
    \label{eq:EShat}
\end{equation}
where $\mathcal{T}_{\alpha}$ collects all quantiles in the estimation grid at or below the threshold $\alpha$. The corresponding change in expected shortfall relative to the baseline is
\begin{equation}
    \Delta\widehat{\text{ES}}_{i,h}^{s}(\alpha) =
    \widehat{\text{ES}}_{i,h}^{s}(\alpha) -
    \widehat{\text{ES}}_{i,h}^{0}(\alpha),
    \label{eq:DeltaES}
\end{equation}
where a negative value indicates that the climate shock worsens the average outcome in the lower tail of the conditional distribution beyond what is captured by any single quantile. We evaluate equations~\eqref{eq:EShat} and~\eqref{eq:DeltaES} at $\alpha = 0.05$, so that $\mathcal{T}_{\alpha}$ includes all quantiles in $\mathcal{T}$ at or below the 5$^{\text{th}}$ percentile.

\begin{table}[h!]
\centering
\caption{Climate growth-at-risk: change in expected shortfall
$\Delta\widehat{\text{ES}}_{i,h}^{s}(0.05)$.}
\label{tab:ES}
\small
\setlength{\tabcolsep}{6pt}
\scalebox{0.9}{
\begin{tabular}{lccclccc}
\toprule
Country & $h=1$ & $h=4$ & $h=10$ & Country & $h=1$ & $h=4$ & $h=10$ \\
\midrule
\multicolumn{4}{l}{\textit{Panel A: Advanced Economies}} &
\multicolumn{4}{l}{\textit{Panel B: Emerging Market Economies}} \\[2pt]
Australia      & $-$0.43 & $-$0.04 & $-$0.07 & Argentina    & $-$0.54 & $-$0.23 & $+$0.09 \\
Austria        & $-$0.33 & $+$0.01 & $-$0.06 & Brazil       & $-$0.54 & $-$0.17 & $+$0.03 \\
Belgium        & $-$0.37 & $-$0.04 & $-$0.03 & Chile        & $-$0.60 & $-$0.39 & $-$0.26 \\
Canada         & $-$0.48 & $-$0.14 & $-$0.05 & China        & $-$0.36 & $+$0.05 & $-$0.04 \\
Finland        & $-$0.44 & $-$0.17 & $-$0.19 & India        & $-$0.36 & $-$0.07 & $-$0.03 \\
France         & $-$0.41 & $-$0.05 & $-$0.05 & Indonesia    & $-$0.21 & $-$0.09 & $+$0.10 \\
Germany        & $-$0.36 & $-$0.08 & $-$0.05 & Korea        & $-$0.60 & $-$0.20 & $-$0.11 \\
Italy          & $-$0.41 & $-$0.10 & $-$0.06 & Malaysia     & $-$0.48 & $-$0.03 & $+$0.08 \\
Japan          & $-$0.45 & $-$0.15 & $-$0.04 & Mexico       & $-$0.76 & $-$0.36 & $-$0.02 \\
Netherlands    & $-$0.40 & $-$0.02 & $+$0.01 & Peru         & $-$0.81 & $-$0.29 & $-$0.04 \\
New Zealand    & $-$0.38 & $-$0.11 & $-$0.08 & Philippines  & $-$0.52 & $-$0.31 & $-$0.23 \\
Norway         & $-$0.40 & $-$0.08 & $+$0.04 & Saudi Arabia & $-$0.40 & $+$0.06 & $-$0.12 \\
Spain          & $-$0.36 & $+$0.01 & $+$0.01 & Singapore    & $-$0.44 & $+$0.09 & $-$0.03 \\
Sweden         & $-$0.31 & $-$0.05 & $-$0.05 & South Africa & $-$0.48 & $-$0.10 & $-$0.12 \\
Switzerland    & $-$0.31 & $-$0.02 & $-$0.08 & Thailand     & $-$0.48 & $-$0.01 & $+$0.16 \\
United Kingdom & $-$0.48 & $-$0.05 & $-$0.10 & Turkey       & $-$0.90 & $-$0.20 & $-$0.17 \\
USA            & $-$0.39 & $-$0.08 & $-$0.06 &              &         &         &         \\
\midrule
\textbf{AE average} & \textbf{$-$0.395} & \textbf{$-$0.068} & \textbf{$-$0.054} &
\textbf{EM average} & \textbf{$-$0.530} & \textbf{$-$0.141} & \textbf{$-$0.044} \\
\midrule
\multicolumn{5}{l}{\textbf{Panel average (All Countries)}} &
\textbf{$-$0.460} & \textbf{$-$0.103} & \textbf{$-$0.049} \\
\bottomrule
\end{tabular}
}
\begin{minipage}{\linewidth}
\smallskip
\footnotesize
\textit{Notes}: Entries report $\Delta\widehat{\text{ES}}_{i,h}^{s}(0.05)$
as defined in equation~\eqref{eq:DeltaES}, evaluated at $\alpha = 0.05$
for horizons $h \in \{1, 4, 10\}$ quarters. A negative value indicates
that a one-standard-deviation climate shock shifts the lower tail of the
conditional output growth distribution downward relative to the baseline,
representing a deterioration in growth-at-risk. Positive values,
indicating an improvement relative to baseline, are prefixed with $+$.
Panel averages are computed as simple cross-country means within each
group and are displayed in bold.
\end{minipage}
\end{table}

Table~\ref{tab:ES} reports $\Delta\widehat{\text{ES}}_{i,h}^{s}(0.05)$ for horizons $h \in \{1, 4, 10\}$ quarters, where a negative entry indicates that a one-standard-deviation climate shock worsens the expected shortfall at the 5 per cent tail of the conditional output growth distribution. Figure~\ref{fig:ESmap} maps the short-horizon ($h=1$) column of the table across the panel. Three findings stand out. First, the impact of climate shocks on growth-at-risk is widespread and uniformly negative at the short horizon ($h = 1$): all 33 countries record a negative $\Delta\widehat{\text{ES}}$, with a panel average of $-$0.460 percentage points. Emerging market economies bear a larger short-run climate tail risk than advanced economies, with EM and AE averages of $-$0.530 and $-$0.395, respectively, consistent with their greater exposure to weather-sensitive sectors. The largest short-run deteriorations are recorded for Turkey ($-$0.90), Peru ($-$0.81), and Mexico ($-$0.76), while Sweden and Switzerland exhibit the smallest impacts among the advanced economies, both at $-$0.31. The map, depicted in Figure \ref{fig:ESmap}, emphasizes this pattern.  The most pronounced deteriorations are concentrated in Latin America, Turkey, and emerging Asia, while continental Europe and the smaller advanced economies form a comparatively lightly shaded band.

\begin{figure}[htbp]
    \centering
    \includegraphics[width=\textwidth]{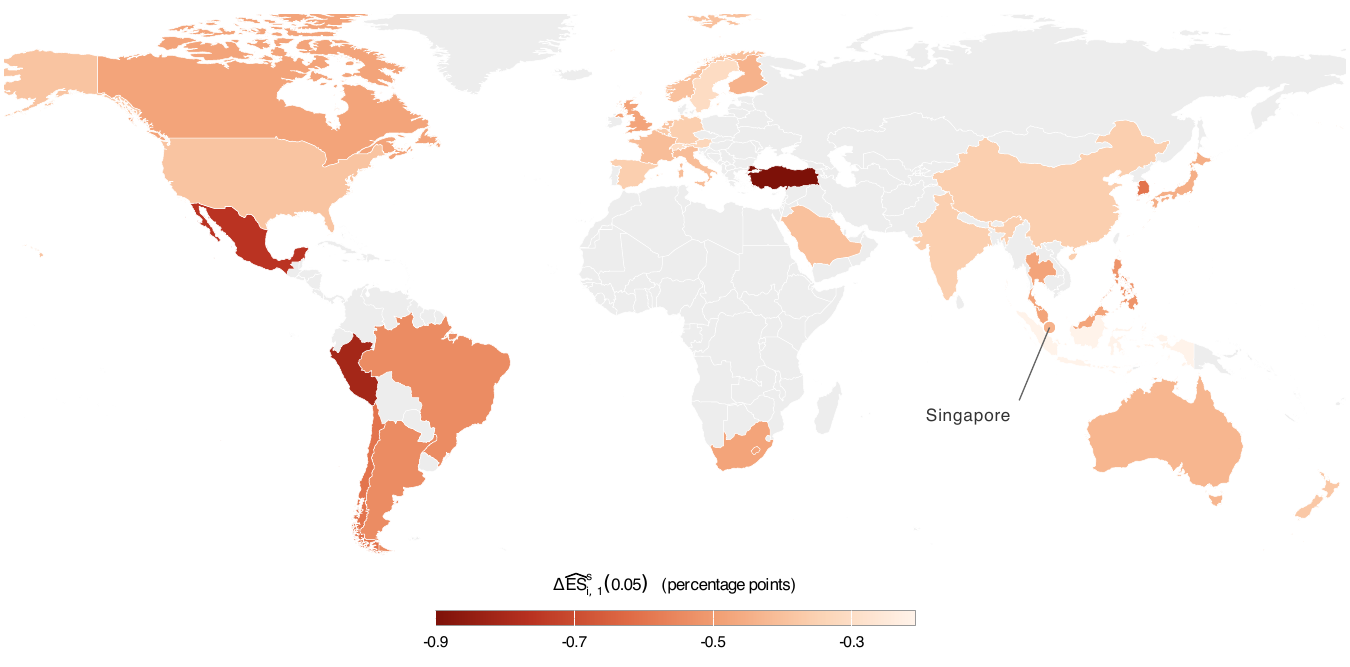}
    \caption*{\scriptsize \textit{Notes}: Each sample country is shaded by its change in expected shortfall $\Delta\widehat{\text{ES}}_{i,h}^{s}(0.05)$ at horizon $h=1$, as defined in equation~\eqref{eq:DeltaES} and reported in the $h=1$ column of Table~\ref{tab:ES}. Darker shading indicates a larger deterioration in the 5 per cent expected shortfall of the conditional output growth distribution following a one-standard-deviation climate shock; all 33 countries record a negative value at this horizon. Countries shaded in light grey are not part of the estimation sample. Singapore is marked with a separate symbol because it is not visible at this scale. Cross-country averages are $-$0.395 for advanced economies, $-$0.530 for emerging market economies, and $-$0.460 for the full panel.}
    \caption{Geography of climate growth-at-risk: change in expected shortfall at $h=1$.}
    \label{fig:ESmap}
\end{figure}

Second, climate tail risks dissipate substantially at medium and longer horizons. The panel-average $\Delta\widehat{\text{ES}}$ falls sharply from $-$0.460 at $h = 1$ to $-$0.103 at $h = 4$ and $-$0.049 at $h = 10$, suggesting that while climate shocks generate significant short-run downside risks, economies partially recover their growth trajectory over time. Third, a small number of economies record positive $\Delta\widehat{\text{ES}}$ at medium and longer horizons --- notably Austria, Spain, China, Saudi Arabia, and Singapore at $h = 4$, and Netherlands, Norway, Spain, Argentina, Brazil, Indonesia, Malaysia, and Thailand at $h = 10$ --- indicating that climate shocks may improve expected shortfall for these economies at longer horizons, potentially reflecting structural features such as sectoral reallocation or terms-of-trade benefits. By contrast, Chile ($-$0.60, $-$0.39, $-$0.26) and the Philippines ($-$0.52, $-$0.31, $-$0.23) are the only economies that sustain large and persistent deteriorations across all three horizons, underscoring their exceptional vulnerability to climate tail risk throughout the forecast horizon.

\subsection{Model Comparison and Out-of-Sample Tail Forecast Performance}
\label{subsec:climate_model_comparison}

In this final section, we evaluate whether the panel quantile model in equation~\eqref{eq:qLP} delivers superior tail risk forecast accuracy relative to a less parameterised alternative. To this end, we conduct a pseudo-out-of-sample forecasting exercise under an expanding window scheme. The initial estimation window covers 1979Q3--1999Q4, and forecasts are evaluated over the period 2000Q1--2023Q3, yielding $T_{\text{oos}} = 95$ out-of-sample quarters and a total evaluation set of $|\mathcal{O}| = N \times T_{\text{oos}}$ country-period pairs. We focus on the one-step-ahead horizon ($h = 1$), at which forecast score differences between competing models are most precisely estimated. At each evaluation period $t$, the model is re-estimated on the expanding sample, and a quantile forecast $\widehat{Q}_{\tau,i,t+1}$ is produced for each country $i$ and each $\tau \in \mathcal{T}$.

We compare our proposed model against a country-specific quantile regression (CSQR) benchmark. The CSQR estimates equation~\eqref{eq:qLP} separately for each country $i$, retaining all four climate regressors $\widehat{T}^{L}_{it}$, $\widehat{T}^{G}_{t}$, $\widehat{\sigma}^{L,\text{RV}}_{it}$, and $\widehat{\sigma}^{G,\text{RV}}_{t}$, but omitting the common time effect $f_{t,h}$. The CSQR is therefore nested within our proposed model. Because both specifications use the same four climate regressors, the comparison isolates the gains from the panel structure itself: the common time effect $f_{t,h}$ and the hierarchical Gaussian-process prior that pools coefficient paths across countries and quantiles. It does not speak to the contribution of the climate regressors, which are common to both models.

Forecast accuracy is evaluated via the quantile scores of \citet{gneiting2011comparing}. For each model, quantile $\tau$, and evaluation set $\mathcal{O}$,
\begin{equation}
    \text{QS}_{\tau} = \frac{1}{|\mathcal{O}|}
    \sum_{(i,t) \in \mathcal{O}}
    \rho_{\tau}\left(y_{i,t+1} -
    \widehat{Q}_{\tau,i,t+1}\right),
    \label{eq:QS}
\end{equation}
where the tick loss function is
\begin{equation}
    \rho_{\tau}(u) = u\left(\tau - \mathbf{1}\{u < 0\}\right),
    \label{eq:tick}
\end{equation}
and $\mathcal{O}$ denotes the out-of-sample evaluation set. We normalise by $|\mathcal{O}|$ to obtain a mean score comparable across models, with a lower value of $\text{QS}_{\tau}$ indicating better forecast accuracy at quantile $\tau$. As our primary interest lies in the left tail of the conditional distribution of output growth, we focus on the lower-tail quantile set
\begin{equation}
    \mathcal{T}_{L} = \{0.01, 0.02, \ldots, 0.10\}
    \subset \mathcal{T},
    \label{eq:tailset}
\end{equation}

and report the mean tail quantile score
\begin{equation}
    \text{QS}_{\text{tail}} = \frac{1}{|\mathcal{T}_{L}|}
    \sum_{\tau \in \mathcal{T}_{L}} \text{QS}_{\tau},
    \label{eq:QStail}
\end{equation}
which aggregates forecast accuracy across the lowest decile of the conditional distribution.

Table~\ref{tab:forecast} reports out-of-sample tail quantile scores expressed relative to the CSQR benchmark, where a value below unity indicates superior growth-at-risk forecast accuracy from our proposed model. The results provide strong and consistent evidence that the proposed panel quantile regression produces more accurate growth-at-risk forecasts than the country-specific quantile regression: the panel-average $\overline{\text{QS}}_{\text{tail}}$ of 0.663 indicates a reduction in tail risk forecast loss of approximately 34 per cent relative to the country-specific benchmark, with every country in the sample recording a $\overline{\text{QS}}_{\text{tail}}$ below unity. Since both models condition on the same four climate regressors, these gains are attributable to the panel structure alone: the common time effect $f_{t,h}$ and the hierarchical prior allow each country to borrow information from the rest of the panel, so that climate-driven tail risks estimated imprecisely in isolation are recovered more accurately under partial pooling.

Forecast improvements are largest at the extreme left tail ($\tau = 0.01$, panel average 0.490), consistent with the view that extreme growth-at-risk events have a strong global climate component that country-specific models cannot adequately capture; at this quantile the largest gains accrue to Spain (0.084), the United Kingdom (0.095), France (0.097), and Mexico (0.100). Averaged across the full tail, the strongest performers are Sweden (0.415), the United Kingdom (0.436), Brazil (0.456), and Mexico (0.464). A small number of country-quantile combinations record scores above unity. At the extreme left tail ($\tau = 0.01$), Germany (1.022) and Thailand (1.246) show modest underperformance, while Australia (2.046) is a notable outlier, more than doubling the benchmark loss at this quantile. At $\tau = 0.05$ and $\tau = 0.10$, a handful of countries --- Canada, Spain, France, and Saudi Arabia --- record scores only marginally above unity. In all cases, however, the $\overline{\text{QS}}_{\text{tail}}$ averaged across the full tail remains below unity, confirming that the panel quantile regression outperforms the country-specific benchmark for every economy in the sample.

\begin{table}[H]
\centering
\caption{Out-of-sample relative tail quantile scores, 2000Q1--2023Q3.}
\label{tab:forecast}
\footnotesize
\setlength{\tabcolsep}{4pt}
\renewcommand{\arraystretch}{0.92}
\resizebox{\textwidth}{!}{%
\begin{tabular}{lcccc @{\hspace{18pt}} lcccc}
\toprule
\multicolumn{5}{c}{\textit{Panel A: Advanced Economies}} &
\multicolumn{5}{c}{\textit{Panel B: Emerging Market Economies}} \\
\cmidrule(lr){1-5}\cmidrule(lr){6-10}
Country & $\tau{=}0.01$ & $\tau{=}0.05$ & $\tau{=}0.10$ & $\overline{\mathrm{QS}}_{\mathrm{tail}}$ &
Country & $\tau{=}0.01$ & $\tau{=}0.05$ & $\tau{=}0.10$ & $\overline{\mathrm{QS}}_{\mathrm{tail}}$ \\
\midrule
Australia      & 2.046 & 0.838 & 0.900 & 0.951 & Argentina    & 0.538 & 0.701 & 0.965 & 0.768 \\
Austria        & 0.671 & 0.571 & 0.544 & 0.633 & Brazil       & 0.241 & 0.438 & 0.642 & 0.456 \\
Belgium        & 0.213 & 0.466 & 0.589 & 0.519 & Chile        & 0.380 & 0.697 & 0.750 & 0.679 \\
Canada         & 0.175 & 1.325 & 1.071 & 0.982 & China        & 0.829 & 0.857 & 0.917 & 0.888 \\
Finland        & 0.454 & 0.569 & 0.592 & 0.546 & India        & 0.780 & 0.597 & 0.672 & 0.641 \\
France         & 0.097 & 0.788 & 1.112 & 0.674 & Indonesia    & 0.277 & 0.770 & 0.799 & 0.716 \\
Germany        & 1.022 & 0.639 & 0.681 & 0.781 & Korea        & 0.322 & 0.557 & 0.642 & 0.550 \\
Italy          & 0.167 & 0.765 & 0.651 & 0.603 & Malaysia     & 0.257 & 0.569 & 0.615 & 0.522 \\
Japan          & 0.459 & 0.698 & 0.649 & 0.637 & Mexico       & 0.100 & 0.585 & 0.496 & 0.464 \\
Netherlands    & 0.324 & 0.417 & 0.613 & 0.495 & Peru         & 0.849 & 0.565 & 0.574 & 0.728 \\
New Zealand    & 0.273 & 0.562 & 0.715 & 0.523 & Philippines  & 0.354 & 0.864 & 0.830 & 0.729 \\
Norway         & 0.266 & 0.541 & 0.662 & 0.522 & Saudi Arabia & 0.396 & 0.932 & 1.044 & 0.855 \\
Spain          & 0.084 & 1.129 & 1.059 & 0.905 & Singapore    & 0.861 & 0.640 & 0.653 & 0.665 \\
Sweden         & 0.255 & 0.391 & 0.505 & 0.415 & South Africa & 0.511 & 0.416 & 0.716 & 0.501 \\
Switzerland    & 0.571 & 0.609 & 0.813 & 0.667 & Thailand     & 1.246 & 0.915 & 0.974 & 0.993 \\
United Kingdom & 0.095 & 0.510 & 0.543 & 0.436 & Turkey       & 0.729 & 0.671 & 0.660 & 0.682 \\
USA            & 0.323 & 0.909 & 0.822 & 0.763 &              &       &       &       &       \\
\midrule
\multicolumn{10}{l}{\textbf{Panel average: $\tau{=}0.01$: 0.490 \quad $\tau{=}0.05$: 0.682 \quad $\tau{=}0.10$: 0.742 \quad $\overline{\mathrm{QS}}_{\mathrm{tail}}$: 0.663}} \\
\bottomrule
\end{tabular}%
}
\begin{minipage}{\linewidth}
\smallskip
\footnotesize
\textit{Notes}: Entries report the ratio
$\mathrm{QS}_{\tau}^{\text{panel}} / \mathrm{QS}_{\tau}^{\text{CSQR}}$ of the panel
model's quantile score to that of the country-specific quantile regression (CSQR)
benchmark, for selected
quantiles $\tau \in \{0.01, 0.05, 0.10\}$ and for the mean tail score
$\overline{\mathrm{QS}}_{\mathrm{tail}}$ averaged over
$\mathcal{T}_{L} = \{0.01, 0.02, \ldots, 0.10\}$, with $\mathrm{QS}_{\tau}$
and $\overline{\mathrm{QS}}_{\mathrm{tail}}$
defined in equations~\eqref{eq:QS} and~\eqref{eq:QStail}. A value
below unity indicates superior tail forecast accuracy from the panel
quantile model. The evaluation period is 2000Q1--2023Q3 under an
expanding window scheme with initial training sample 1979Q3--1999Q4.
\end{minipage}
\end{table}

\section{Robustness}
\label{subsec:climate_robustness}

In this section, we undertake three robustness checks to assess the sensitivity of our empirical results to key modeling choices, focusing on the role of the common factor $f_{t,h}$, the local/global decomposition of climate shocks, and the stability of the AE/EM heterogeneity documented above. The first check re-estimates the panel quantile model with the common factor $f_{t,h}$ removed entirely, isolating how much of the baseline climate responses this component accounts for. Figure~\ref{fig:AverageIRFsrobust1} plots the resulting panel-average DIRFs. The two local shocks are essentially unaffected, both in timing and magnitude, relative to the baseline in Figure~\ref{fig:AverageIRFs}. The two global shocks retain the same sign and horizon profile as in the baseline but are noticeably smaller in magnitude once the common factor is removed. This is consistent with the identification argument in Section~\ref{subsec:climate_specification}: since the global climate regressors are collinear with any unrestricted common time variation, $f_{t,h}$ is what allows the model to recover the full magnitude of the average global climate effect, separately from the country-level heterogeneity. The attenuation of the global responses in this restricted specification therefore corroborates, rather than undermines, the role of the common factor in the baseline model.

\begin{figure}[htbp]
    \centering
    \includegraphics[width=\textwidth]{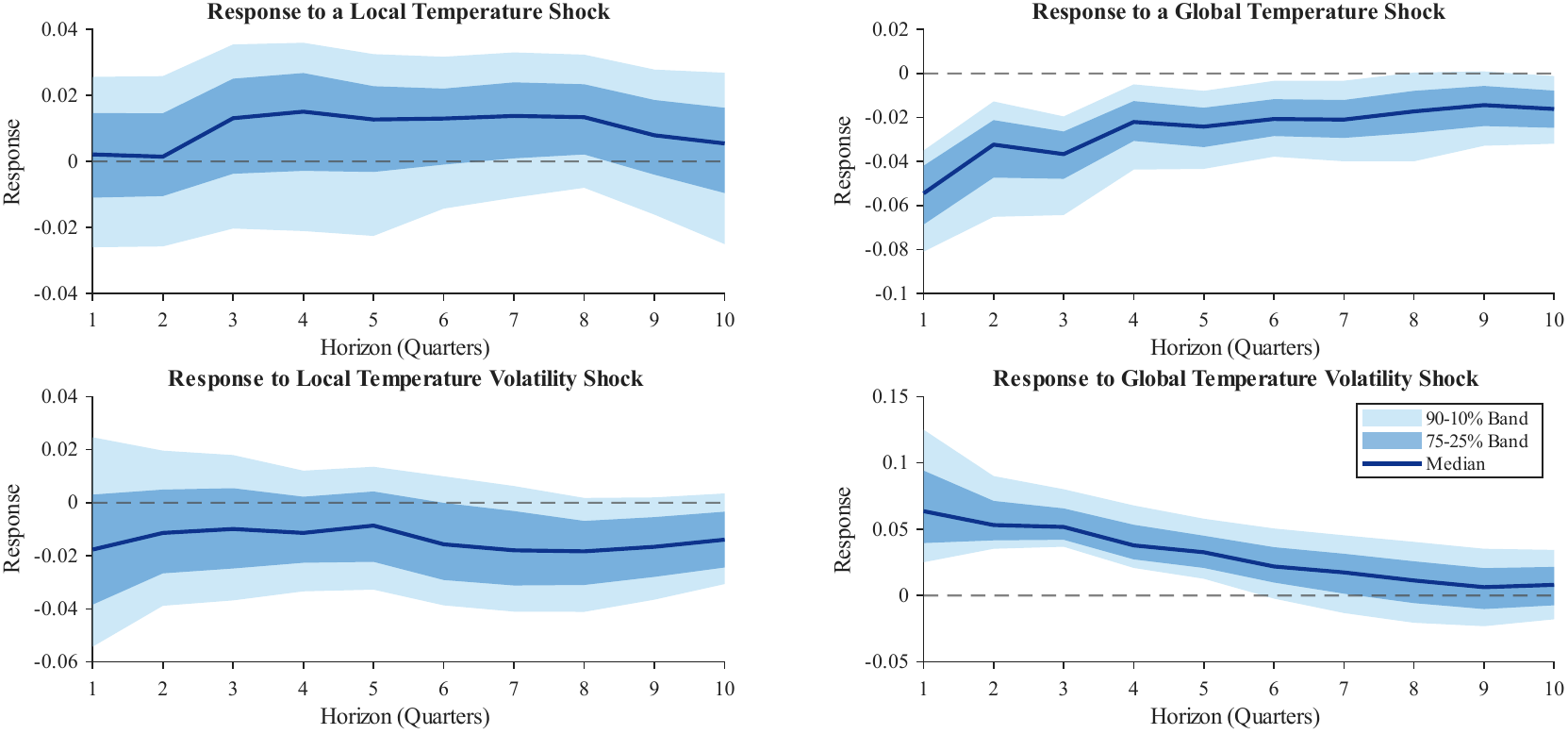}
    \caption*{\scriptsize \textit{Notes}: Each panel plots $\overline{\text{DIRF}}^{\,s}_{h}(\tau)$ as defined in equation~\eqref{eq:avgDIRF} for $h = 1, \ldots, 10$ quarters, estimated from the panel quantile model in equation~\eqref{eq:qLP} with the common factor $f_{t,h}$ removed. The solid line is the median response ($\tau = 0.50$). Dark and light shaded bands correspond to the 25$^{\text{th}}$--75$^{\text{th}}$ and 10$^{\text{th}}$--90$^{\text{th}}$ quantile responses, respectively. Panels report responses to: (a) local temperature shock; (b) global temperature shock; (c) local temperature volatility shock; (d) global temperature volatility shock.}
    \caption{Panel-average distributional impulse responses to local and global climate shocks with the common factor removed ($f_{t,h} = 0$).}
    \label{fig:AverageIRFsrobust1}
\end{figure}

We further evaluate this restricted specification along two additional dimensions: the stability of its expected shortfall estimates and its out-of-sample forecasting performance. In terms of expected shortfall, estimates from the restricted model are considerably less stable across horizons than the baseline: the panel-average deterioration roughly doubles at $h=1$ and reverses sign at $h=4$ (Table~\ref{tab:ES_robust1} in the appendix). This instability provides further evidence that the common factor contributes to producing well-behaved tail-risk estimates.

We next compare the out-of-sample forecasting tail risk performance of our baseline model to that of the restricted specification, following the procedure described in Section~\ref{subsec:climate_model_comparison}. Table~\ref{tab:forecast_robust1} in the appendix reports the baseline's quantile scores relative to the restricted model, where a value below unity indicates that the baseline yields more accurate tail risk forecasts. The baseline outperforms the restricted specification for every country in the sample except Saudi Arabia and Thailand, where the two models perform comparably. The panel-average relative score is 0.70, indicating a reduction in tail risk forecast loss of approximately 30 per cent when the common factor is included, with the largest gains concentrated at the extreme left tail ($\tau=0.01$, panel average 0.55). Taken together, these results indicate that the common factor $f_{t,h}$ materially improves both the stability and the out-of-sample accuracy of the model's climate-driven tail-risk estimates, rather than serving as a purely specification-driven addition.

Next, for the second robustness check, we estimate a restricted version of the model in which the global climate regressors are excluded ($\widehat{T}^{G}_{t} = \widehat{\sigma}^{G,\text{RV}}_{t} = 0$), retaining only the local temperature and local volatility shocks. Figure~\ref{fig:AverageIRFsrobust2} reports the resulting panel-average DIRFs. Both the local temperature and local temperature volatility responses are close to the baseline in Figure~\ref{fig:AverageIRFs}, in both timing and magnitude. This indicates that the local shock responses are not an artefact of jointly estimating the global climate regressors, and that the local and global components of the model are identified largely independently of one another.

\begin{figure}[htbp]
    \centering
    \includegraphics[width=\textwidth]{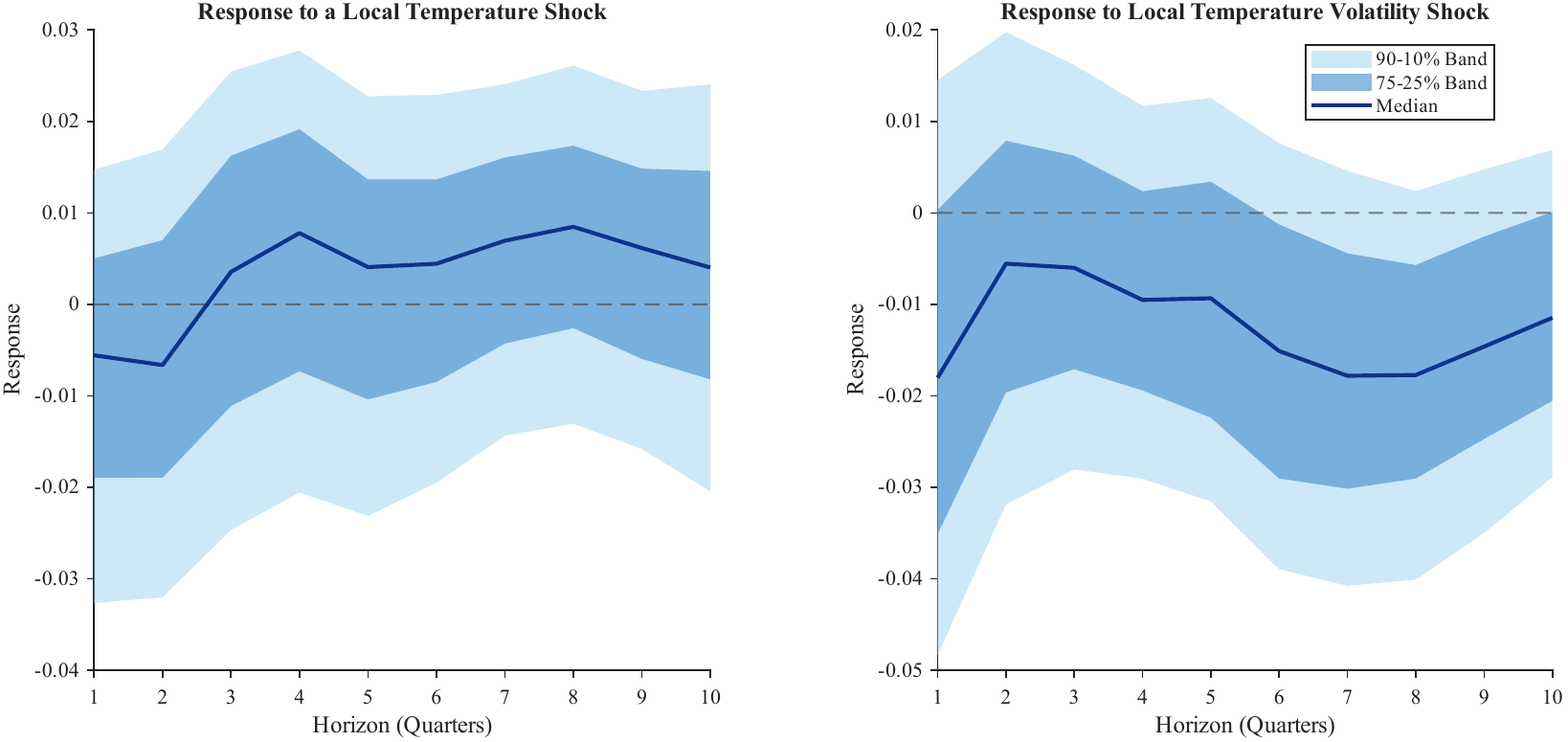}
    \caption*{\scriptsize \textit{Notes}: Each panel plots $\overline{\text{DIRF}}^{\,s}_{h}(\tau)$ as defined in equation~\eqref{eq:avgDIRF} for $h = 1, \ldots, 10$ quarters, estimated from the panel quantile model in equation~\eqref{eq:qLP} with the global temperature and global volatility regressors removed. The solid line is the median response ($\tau = 0.50$). Dark and light shaded bands correspond to the 25$^{\text{th}}$--75$^{\text{th}}$ and 10$^{\text{th}}$--90$^{\text{th}}$ quantile responses, respectively. Compare to Figure~\ref{fig:AverageIRFs} for the baseline specification.}
    \caption{Panel-average distributional impulse responses to local climate shocks with the global climate regressors removed ($\widehat{T}^{G}_{t} = \widehat{\sigma}^{G,\text{RV}}_{t} = 0$).}
    \label{fig:AverageIRFsrobust2}
\end{figure}

\begin{figure}[H]
    \centering
    \includegraphics[width=\textwidth]{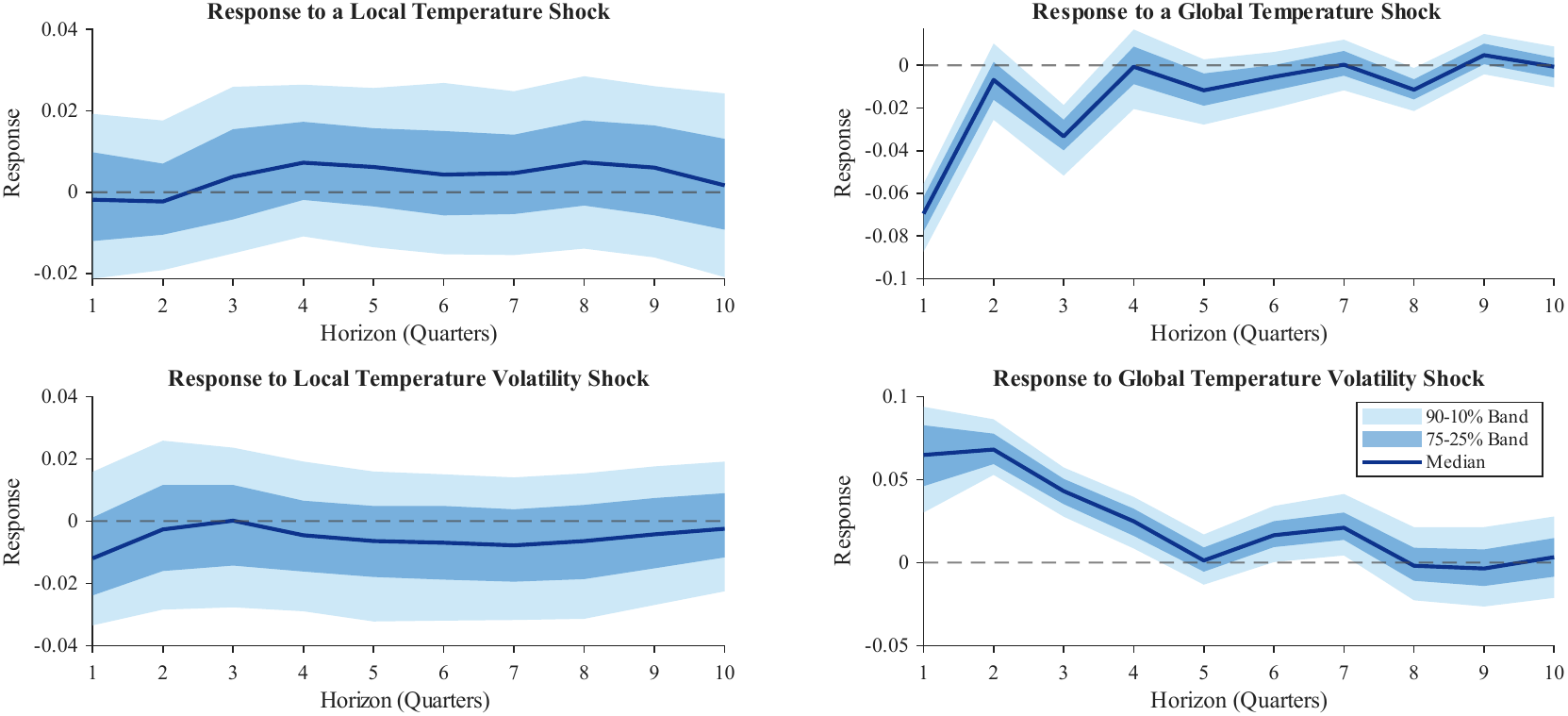}
    \caption*{\scriptsize \textit{Notes}: Each panel plots $\overline{\text{DIRF}}^{\,s}_{h}(\tau)$ as defined in equation~\eqref{eq:avgDIRF} for $h = 1, \ldots, 10$ quarters, estimated from the panel quantile model in equation~\eqref{eq:qLP} on the advanced economy subsample only. The solid line is the median response ($\tau = 0.50$). Dark and light shaded bands correspond to the 25$^{\text{th}}$--75$^{\text{th}}$ and 10$^{\text{th}}$--90$^{\text{th}}$ quantile responses, respectively. Panels report responses to: (a) local temperature shock; (b) global temperature shock; (c) local temperature volatility shock; (d) global temperature volatility shock.}
    \caption{Panel-average distributional impulse responses to local and global climate shocks, advanced economies only.}
    \label{fig:AverageIRFsrobust3}
\end{figure}

Finally, in the last robustness check, we estimate the panel quantile model separately on the advanced- and emerging-economy subsamples. This tests whether the AE/EM heterogeneity documented in Section~\ref{subsec:heterogeneity} persists when each group is estimated entirely independently, with its own dynamic common factor $f_{t,h}$ and its own local and global climate regressors, rather than emerging only under partial pooling across a single global hierarchical prior. Figures~\ref{fig:AverageIRFsrobust3} and~\ref{fig:AverageIRFsrobust4} report the resulting panel-average DIRFs for the AE and EM subsamples, respectively. The overall dynamics of the responses are broadly similar to the pooled baseline in Figure~\ref{fig:AverageIRFs}, both in sign and horizon profile, across all four climate shocks. This indicates that our baseline results are not driven by pooling advanced and emerging economies under a single hierarchical prior.

\begin{figure}[H]
    \centering
    \includegraphics[width=\textwidth]{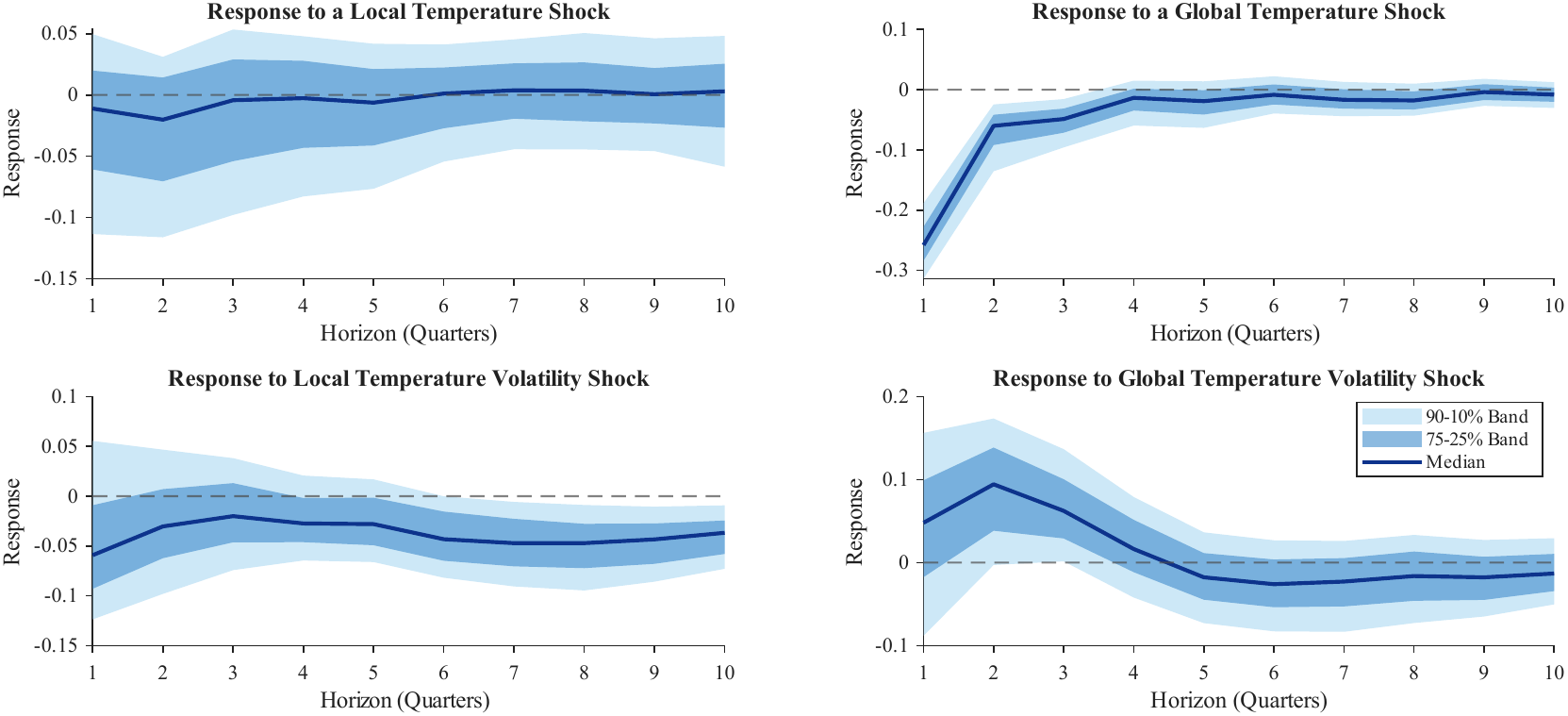}
    \caption*{\scriptsize \textit{Notes}: Each panel plots $\overline{\text{DIRF}}^{\,s}_{h}(\tau)$ as defined in equation~\eqref{eq:avgDIRF} for $h = 1, \ldots, 10$ quarters, estimated from the panel quantile model in equation~\eqref{eq:qLP} on the emerging market economy subsample only. The solid line is the median response ($\tau = 0.50$). Dark and light shaded bands correspond to the 25$^{\text{th}}$--75$^{\text{th}}$ and 10$^{\text{th}}$--90$^{\text{th}}$ quantile responses, respectively. Panels report responses to: (a) local temperature shock; (b) global temperature shock; (c) local temperature volatility shock; (d) global temperature volatility shock.}
    \caption{Panel-average distributional impulse responses to local and global climate shocks, emerging economies only.}
    \label{fig:AverageIRFsrobust4}
\end{figure}

\section{Conclusion}
\label{sec:conclusion}
This paper developed a hierarchical Bayesian panel quantile regression framework, combining Gaussian-process-smoothed coefficient paths with Bernstein-polynomial monotonicity constraints to address quantile crossing, and applied it to characterize the distributional macroeconomic effects of climate shocks across a 33-country quarterly panel spanning 1979--2023. We find that global temperature shocks generate a systemic downside risk to output growth shared broadly across advanced and emerging economies, while global temperature volatility shocks generate a short-lived positive response instead. Both stand in contrast to local, idiosyncratic shocks, whose effects are modest on average but substantially heterogeneous across countries.

Combining the effects of all four climate shocks, our climate growth-at-risk measure shows that downside risk is pervasive and uniformly negative in the immediate aftermath of a shock, disproportionately affecting emerging market economies, before dissipating substantially at longer horizons. In an out-of-sample forecasting exercise, our proposed panel quantile model reduces tail-risk forecast loss by approximately 34 per cent relative to a country-specific benchmark, underscoring the practical value of pooling information across countries when forecasting climate-driven downside risk.

\newpage
\bibliography{ref}

\newpage
\appendix

\section{Proofs of Theoretical Results}
We first present the proof of Lemma~\ref{lem:noncrossing_sufficient}.
\begin{proof}
Fix $0<\tau<\tau'<1$. The difference between the two conditional
quantiles is
\[
\begin{aligned}
Q_{\tau'}(y_{it}\mid \mathcal I_{it})
-
Q_{\tau}(y_{it}\mid \mathcal I_{it})
&=
\left[
a_i(\tau')+f_t+\bm{z}_{it}'\bm{\beta}_i(\tau')
\right]
-
\left[
a_i(\tau)+f_t+\bm{z}_{it}'\bm{\beta}_i(\tau)
\right]  \\
&=
a_i(\tau')-a_i(\tau)
+
\bm{z}_{it}'\left\{\bm{\beta}_i(\tau')-\bm{\beta}_i(\tau)\right\}.
\end{aligned}
\]
The common time effect cancels because it is invariant across quantiles.
By assumption,
\[
a_i(\tau')-a_i(\tau)\ge 0
\]
and, for every $k=1,\ldots,K$,
\[
\beta_{ik}(\tau')-\beta_{ik}(\tau)\ge 0.
\]
Since $\bm{z}_{it}\in[0,1]^K$, each component satisfies $z_{itk}\ge 0$.
Therefore,
\[
\bm{z}_{it}'\left\{\bm{\beta}_i(\tau')-\bm{\beta}_i(\tau)\right\}
=
\sum_{k=1}^K
z_{itk}
\left\{
\beta_{ik}(\tau')-\beta_{ik}(\tau)
\right\}
\ge 0.
\]
Thus,
\[
Q_{\tau'}(y_{it}\mid \mathcal I_{it})
-
Q_{\tau}(y_{it}\mid \mathcal I_{it})
\ge 0,
\]
which proves the result.
\end{proof}
We shall now present the proof of Theorem~\ref{thm:soft_noncrossing}.
\begin{proof}
Fix $0<\tau<\tau'<1$. For any $\bm{z}\in\mathcal Z$, the difference between the
two fitted conditional quantiles is
\[
\begin{aligned}
Q_{\tau'}(y_{it}\mid \mathcal I_{it})
-
Q_{\tau}(y_{it}\mid \mathcal I_{it})
&=
\bm{r}(\bm{z})'
\left\{
\bm{\theta}_i(\tau')-\bm{\theta}_i(\tau)
\right\}  \\
&=
\bm{r}(\bm{z})'
\left\{
\bm{\mu}(\tau')-\bm{\mu}(\tau)
\right\}
+
\bm{r}(\bm{z})'\bm{d}_i(\tau,\tau'),
\end{aligned}
\]
where
\[
\bm{d}_i(\tau,\tau')=\bm{u}_i(\tau')-\bm{u}_i(\tau).
\]
By definition of the population margin,
\[
\bm{r}(\bm{z})'
\left\{
\bm{\mu}(\tau')-\bm{\mu}(\tau)
\right\}
\ge
m_{\tau,\tau'}
\qquad
\text{for all } \bm{z}\in\mathcal Z.
\]
Therefore, crossing can occur only if there exists some $\bm{z}\in\mathcal Z$
such that
\[
\bm{r}(\bm{z})'\bm{d}_i(\tau,\tau')<-m_{\tau,\tau'}.
\]

Let $d_{ij}(\tau,\tau')$ denote the $j$th component of
$\bm{d}_i(\tau,\tau')$, for $j=1,\ldots,p$. Since
$\bm{r}(\bm{z})=(1,\bm{z}')'$ and $\bm{z}\in[0,1]^K$, all components of $\bm{r}(\bm{z})$ are nonnegative
and
\[
\sum_{j=1}^p r_j(\bm{z})\le p.
\]
If
\[
d_{ij}(\tau,\tau')\ge -\frac{m_{\tau,\tau'}}{p}
\qquad
\text{for all } j=1,\ldots,p,
\]
then
\[
\bm{r}(\bm{z})'\bm{d}_i(\tau,\tau')
=
\sum_{j=1}^p r_j(\bm{z})d_{ij}(\tau,\tau')
\ge
-\frac{m_{\tau,\tau'}}{p}
\sum_{j=1}^p r_j(\bm{z})
\ge
-m_{\tau,\tau'}.
\]
Hence crossing implies that at least one component satisfies
\[
d_{ij}(\tau,\tau')
<
-\frac{m_{\tau,\tau'}}{p}.
\]
By the union bound,
\[
\Pr(\text{crossing})
\le
\sum_{j=1}^p
\Pr\left(
d_{ij}(\tau,\tau')
<
-\frac{m_{\tau,\tau'}}{p}
\right).
\]
Under Assumption~\ref{ass:gp}, each $d_{ij}(\tau,\tau')$ is mean-zero
Gaussian with variance $\operatorname{Var}(d_{ij}(\tau,\tau'))
=2\{1-K_\lambda(\tau,\tau')\}\Sigma_{jj}\le\bar\nu^2_{\tau,\tau'}$.
Using the standard Gaussian tail bound,
\[
\Pr\left(
d_{ij}(\tau,\tau')
<
-\frac{m_{\tau,\tau'}}{p}
\right)
\le
\exp\left(
-
\frac{m_{\tau,\tau'}^2}
{2p^2\operatorname{Var}(d_{ij}(\tau,\tau'))}
\right)
\le
\exp\left(
-
\frac{m_{\tau,\tau'}^2}
{2p^2\bar\nu^2_{\tau,\tau'}}
\right).
\]
Summing over $j=1,\ldots,p$ gives
\[
\Pr(\text{crossing})
\le
p
\exp\left(
-
\frac{m_{\tau,\tau'}^2}
{2p^2\bar\nu^2_{\tau,\tau'}}
\right).
\]
For the exponential kernel,
\[
\bar\nu^2_{\tau,\tau'}
=
2\{1-\exp(-|\tau'-\tau|/\lambda)\}\max_{1\le j\le p}\Sigma_{jj},
\]
which yields the stated bound.
\end{proof}
\begin{corollary}
\label{cor:grid_soft_noncrossing}
Let $0<\tau_1<\cdots<\tau_L<1$ be the finite grid of quantile indices
used in estimation, and let
\[
\bm{r}(\bm{z})=(1,\bm{z}')'\in\mathbb{R}^{p},
\qquad \bm{z}\in\mathcal Z\subseteq[0,1]^K,
\qquad p=K+1.
\]
For each adjacent pair $(\tau_\ell,\tau_{\ell+1})$, define the population
margin
\[
m_\ell
=
\inf_{\bm{z}\in\mathcal Z}
\bm{r}(\bm{z})'
\left\{
\bm{\mu}(\tau_{\ell+1})-\bm{\mu}(\tau_\ell)
\right\},
\qquad \ell=1,\ldots,L-1.
\]
Suppose that
\[
m_*=\min_{\ell=1,\ldots,L-1}m_\ell>0.
\]
Then, for a given unit $i$,
\[
\Pr\left(
\exists \ell\in\{1,\ldots,L-1\},\
\exists \bm{z}\in\mathcal Z:
Q_{\tau_{\ell+1}}(y_{it}\mid\mathcal I_{it})
<
Q_{\tau_\ell}(y_{it}\mid\mathcal I_{it})
\right)
\le
p
\sum_{\ell=1}^{L-1}
\exp\left(
-
\frac{m_\ell^2}
{2p^2\bar\nu_\ell^2}
\right),
\]
where $p=K+1$ and
\[
\bar\nu_\ell^2
=
2\{1-K_\lambda(\tau_{\ell+1},\tau_\ell)\}\max_{1\le j\le p}\Sigma_{jj}.
\]
In particular, for the exponential kernel
\[
K_\lambda(\tau,\tau')
=
\exp\left(-\frac{|\tau-\tau'|}{\lambda}\right),
\]
we have
\[
\Pr\left(
\exists \ell,\exists \bm{z}:
Q_{\tau_{\ell+1}}(y_{it}\mid\mathcal I_{it})
<
Q_{\tau_\ell}(y_{it}\mid\mathcal I_{it})
\right)
\le
p
\sum_{\ell=1}^{L-1}
\exp\left[
-
\frac{m_\ell^2}
{
4p^2
\max_{1\le j\le p}\Sigma_{jj}
\left\{
1-\exp(-|\tau_{\ell+1}-\tau_\ell|/\lambda)
\right\}
}
\right].
\]
Consequently, if the adjacent population margins are uniformly positive,
that is $m_*>0$, then the probability of a crossing on the finite quantile
grid converges to zero as $\max_j\Sigma_{jj}\to0$.
\end{corollary}
\begin{proof}
For each adjacent pair $(\tau_\ell,\tau_{\ell+1})$, Theorem
\ref{thm:soft_noncrossing} gives
\[
\Pr\left(
\exists \bm{z}\in\mathcal Z:
Q_{\tau_{\ell+1}}(y_{it}\mid\mathcal I_{it})
<
Q_{\tau_\ell}(y_{it}\mid\mathcal I_{it})
\right)
\le
p
\exp\left(
-
\frac{m_\ell^2}
{2p^2\bar\nu_\ell^2}
\right).
\]
Taking a union bound over the adjacent pairs
$\ell=1,\ldots,L-1$ yields
\[
\Pr\left(
\exists \ell,\exists \bm{z}:
Q_{\tau_{\ell+1}}(y_{it}\mid\mathcal I_{it})
<
Q_{\tau_\ell}(y_{it}\mid\mathcal I_{it})
\right)
\le
p
\sum_{\ell=1}^{L-1}
\exp\left(
-
\frac{m_\ell^2}
{2p^2\bar\nu_\ell^2}
\right).
\]
For the exponential kernel,
\[
\bar\nu_\ell^2
=
2\{1-\exp(-|\tau_{\ell+1}-\tau_\ell|/\lambda)\}\max_{1\le j\le p}\Sigma_{jj},
\]
which gives the stated expression.

Finally, if $m_*=\min_{\ell}m_\ell>0$, then each term in the bound converges
to zero as $\max_j\Sigma_{jj}\to0$. Since the grid is finite, the whole sum converges
to zero. This proves the result.
\end{proof}

\section{Data Appendix}

\begin{table}[H]
\centering
\caption{List of Countries by Region}
\label{tab:countries}
\renewcommand{\arraystretch}{1.2}
\setlength{\tabcolsep}{6pt}
\begin{tabular}{p{5cm} p{5cm}}
\toprule
\textbf{Region} & \textbf{Countries} \\
\midrule
\textit{Asia and Pacific}
& Australia, China, India, Indonesia, Japan, Korea, Malaysia, New Zealand, Philippines, Singapore, Thailand \\[6pt]

\textit{North America}
& Canada, Mexico, United States \\[6pt]

\textit{South America}
& Argentina, Brazil, Chile, Peru \\[6pt]

\textit{Europe}
& Austria, Belgium, Finland, France, Germany, Italy, Netherlands, Norway, Spain, Sweden, Switzerland, Turkey, United Kingdom \\[6pt]

\textit{Middle East and Africa}
& Saudi Arabia, South Africa \\
\bottomrule
\end{tabular}
\end{table}

\section{Additional Robustness Results}

\begin{table}[H]
\centering
\caption{Robustness: change in expected shortfall $\Delta\widehat{\text{ES}}_{i,h}^{s}(0.05)$ with and without the common factor.}
\label{tab:ES_robust1}
\resizebox{\textwidth}{!}{%
\scriptsize
\setlength{\tabcolsep}{3pt}
\begin{tabular}{lccc ccc ccc}
\toprule
& \multicolumn{3}{c}{$h=1$} & \multicolumn{3}{c}{$h=4$} & \multicolumn{3}{c}{$h=10$} \\
\cmidrule(lr){2-4} \cmidrule(lr){5-7} \cmidrule(lr){8-10}
Country & No factor & Baseline & $\Delta$ & No factor & Baseline & $\Delta$ & No factor & Baseline & $\Delta$ \\
\midrule
Argentina      & $-$1.43 & $-$0.54 & $-$0.89 & $+$0.02 & $-$0.23 & $+$0.25 & $+$0.15 & $+$0.09 & $+$0.06 \\
Australia      & $-$0.23 & $-$0.43 & $+$0.20 & $+$0.07 & $-$0.04 & $+$0.11 & $-$0.07 & $-$0.07 & $+$0.00 \\
Austria        & $-$0.48 & $-$0.33 & $-$0.15 & $+$0.08 & $+$0.01 & $+$0.07 & $-$0.06 & $-$0.06 & $+$0.00 \\
Belgium        & $-$0.48 & $-$0.37 & $-$0.11 & $+$0.05 & $-$0.04 & $+$0.09 & $-$0.05 & $-$0.03 & $-$0.02 \\
Brazil         & $-$0.75 & $-$0.54 & $-$0.21 & $+$0.11 & $-$0.17 & $+$0.28 & $+$0.05 & $+$0.03 & $+$0.02 \\
Canada         & $-$0.38 & $-$0.48 & $+$0.10 & $+$0.21 & $-$0.14 & $+$0.35 & $-$0.06 & $-$0.05 & $-$0.01 \\
China          & $-$0.01 & $-$0.36 & $+$0.35 & $+$0.11 & $+$0.05 & $+$0.06 & $+$0.03 & $-$0.04 & $+$0.07 \\
Chile          & $-$1.94 & $-$0.60 & $-$1.34 & $-$0.10 & $-$0.39 & $+$0.29 & $-$0.20 & $-$0.26 & $+$0.06 \\
Finland        & $-$0.38 & $-$0.44 & $+$0.06 & $+$0.22 & $-$0.17 & $+$0.39 & $-$0.20 & $-$0.19 & $-$0.01 \\
France         & $-$0.79 & $-$0.41 & $-$0.38 & $-$0.09 & $-$0.05 & $-$0.04 & $-$0.09 & $-$0.05 & $-$0.04 \\
Germany        & $-$0.19 & $-$0.36 & $+$0.17 & $+$0.27 & $-$0.08 & $+$0.35 & $+$0.00 & $-$0.05 & $+$0.05 \\
India          & $-$0.51 & $-$0.36 & $-$0.15 & $+$0.03 & $-$0.07 & $+$0.10 & $+$0.02 & $-$0.03 & $+$0.05 \\
Indonesia      & $-$0.61 & $-$0.21 & $-$0.40 & $+$0.25 & $-$0.09 & $+$0.34 & $+$0.16 & $+$0.10 & $+$0.06 \\
Italy          & $-$0.12 & $-$0.41 & $+$0.29 & $+$0.29 & $-$0.10 & $+$0.39 & $-$0.08 & $-$0.06 & $-$0.02 \\
Japan          & $-$0.86 & $-$0.45 & $-$0.41 & $+$0.29 & $-$0.15 & $+$0.44 & $+$0.03 & $-$0.04 & $+$0.07 \\
Korea          & $-$0.62 & $-$0.60 & $-$0.02 & $+$0.14 & $-$0.20 & $+$0.34 & $-$0.10 & $-$0.11 & $+$0.01 \\
Malaysia       & $-$0.93 & $-$0.48 & $-$0.45 & $+$0.28 & $-$0.03 & $+$0.31 & $+$0.17 & $+$0.08 & $+$0.09 \\
Mexico         & $-$1.48 & $-$0.76 & $-$0.72 & $+$0.16 & $-$0.36 & $+$0.52 & $+$0.10 & $-$0.02 & $+$0.12 \\
Netherlands    & $-$0.67 & $-$0.40 & $-$0.27 & $+$0.25 & $-$0.02 & $+$0.27 & $+$0.03 & $+$0.01 & $+$0.02 \\
Norway         & $-$0.49 & $-$0.40 & $-$0.09 & $+$0.19 & $-$0.08 & $+$0.27 & $+$0.01 & $+$0.04 & $-$0.03 \\
New Zealand    & $-$0.36 & $-$0.38 & $+$0.02 & $+$0.15 & $-$0.11 & $+$0.26 & $-$0.01 & $-$0.08 & $+$0.07 \\
Peru           & $-$1.29 & $-$0.81 & $-$0.48 & $+$0.01 & $-$0.29 & $+$0.30 & $+$0.03 & $-$0.04 & $+$0.07 \\
Philippines    & $-$1.31 & $-$0.52 & $-$0.79 & $+$0.22 & $-$0.31 & $+$0.53 & $-$0.14 & $-$0.23 & $+$0.09 \\
South Africa   & $-$0.49 & $-$0.48 & $-$0.01 & $+$0.10 & $-$0.10 & $+$0.20 & $-$0.21 & $-$0.12 & $-$0.09 \\
Saudi Arabia   & $-$0.70 & $-$0.40 & $-$0.30 & $+$0.41 & $+$0.06 & $+$0.35 & $-$0.08 & $-$0.12 & $+$0.04 \\
Singapore      & $-$1.07 & $-$0.44 & $-$0.63 & $+$0.39 & $+$0.09 & $+$0.30 & $+$0.00 & $-$0.03 & $+$0.03 \\
Spain          & $+$0.04 & $-$0.36 & $+$0.40 & $+$0.33 & $+$0.01 & $+$0.32 & $+$0.08 & $+$0.01 & $+$0.07 \\
Sweden         & $-$0.62 & $-$0.31 & $-$0.31 & $+$0.30 & $-$0.05 & $+$0.35 & $+$0.02 & $-$0.05 & $+$0.07 \\
Switzerland    & $-$0.58 & $-$0.31 & $-$0.27 & $+$0.03 & $-$0.02 & $+$0.05 & $-$0.04 & $-$0.08 & $+$0.04 \\
Thailand       & $-$1.29 & $-$0.48 & $-$0.81 & $+$0.27 & $-$0.01 & $+$0.28 & $+$0.20 & $+$0.16 & $+$0.04 \\
Turkey         & $-$2.21 & $-$0.90 & $-$1.31 & $+$0.12 & $-$0.20 & $+$0.32 & $-$0.09 & $-$0.17 & $+$0.08 \\
United Kingdom & $-$0.24 & $-$0.48 & $+$0.24 & $+$0.37 & $-$0.05 & $+$0.42 & $-$0.11 & $-$0.10 & $-$0.01 \\
USA            & $-$0.35 & $-$0.39 & $+$0.04 & $+$0.14 & $-$0.08 & $+$0.22 & $-$0.02 & $-$0.06 & $+$0.04 \\
\midrule
\textbf{Panel average} & \textbf{$-$0.722} & \textbf{$-$0.460} & \textbf{$-$0.262} & \textbf{$+$0.172} & \textbf{$-$0.103} & \textbf{$+$0.275} & \textbf{$-$0.016} & \textbf{$-$0.049} & \textbf{$+$0.033} \\
\bottomrule
\end{tabular}%
}
\begin{minipage}{\linewidth}
\smallskip
\footnotesize
\textit{Notes}: Entries report $\Delta\widehat{\text{ES}}_{i,h}^{s}(0.05)$
as defined in equation~\eqref{eq:DeltaES}, evaluated at $\alpha=0.05$
for horizons $h \in \{1,4,10\}$ quarters, comparing the baseline
panel quantile model to a restricted specification with the common
factor $f_{t,h}$ removed entirely. In the ``No factor'' and
``Baseline'' columns, a negative value indicates that a one-standard-deviation
climate shock shifts the lower tail of the conditional output growth distribution
downward, representing a deterioration in growth-at-risk. The $\Delta$ column reports
the difference between the two specifications (no factor minus baseline) and is
therefore not itself a growth-at-risk measure: a negative entry means the restricted
model attributes a larger deterioration to the shock than the baseline does. Panel averages are computed as simple cross-country
means and are displayed in bold.
\end{minipage}
\end{table}

\begin{table}[H]
\centering
\caption{Robustness: out-of-sample relative tail quantile scores, baseline versus restricted model without the common factor.}
\label{tab:forecast_robust1}
\footnotesize
\setlength{\tabcolsep}{4pt}
\renewcommand{\arraystretch}{0.92}
\resizebox{\textwidth}{!}{%
\begin{tabular}{lcccc @{\hspace{18pt}} lcccc}
\toprule
& \multicolumn{3}{c}{Quantile $\tau$} & & & \multicolumn{3}{c}{Quantile $\tau$} & \\
\cmidrule(lr){2-4}\cmidrule(lr){7-9}
Country & $\tau{=}0.01$ & $\tau{=}0.05$ & $\tau{=}0.10$ & Average &
Country & $\tau{=}0.01$ & $\tau{=}0.05$ & $\tau{=}0.10$ & Average \\
\midrule
Argentina      & 0.89 & 0.90 & 0.94 & 0.91 & New Zealand    & 0.42 & 0.58 & 0.73 & 0.59 \\
Australia      & 0.51 & 0.71 & 0.86 & 0.72 & Peru           & 0.35 & 0.53 & 0.59 & 0.51 \\
Austria        & 0.78 & 0.79 & 0.73 & 0.77 & Philippines    & 0.66 & 0.88 & 0.77 & 0.82 \\
Belgium        & 0.41 & 0.56 & 0.61 & 0.55 & South Africa   & 0.24 & 0.47 & 0.65 & 0.47 \\
Brazil         & 0.39 & 0.61 & 0.73 & 0.60 & Saudi Arabia   & 0.92 & 0.96 & 1.07 & 0.97 \\
Canada         & 0.48 & 0.77 & 0.84 & 0.74 & Singapore      & 0.39 & 0.71 & 0.79 & 0.67 \\
China          & 0.72 & 0.88 & 0.99 & 0.88 & Spain          & 0.55 & 0.74 & 0.79 & 0.73 \\
Chile          & 0.49 & 0.76 & 0.80 & 0.72 & Sweden         & 0.34 & 0.56 & 0.65 & 0.54 \\
Finland        & 0.53 & 0.65 & 0.73 & 0.65 & Switzerland    & 0.51 & 0.70 & 0.85 & 0.71 \\
France         & 0.51 & 0.78 & 0.77 & 0.74 & Thailand       & 1.08 & 0.96 & 0.94 & 0.98 \\
Germany        & 0.46 & 0.62 & 0.65 & 0.60 & Turkey         & 0.63 & 0.68 & 0.76 & 0.70 \\
India          & 0.51 & 0.66 & 0.75 & 0.66 & United Kingdom & 0.45 & 0.62 & 0.66 & 0.61 \\
Indonesia      & 0.70 & 0.80 & 0.87 & 0.80 & USA            & 0.42 & 0.70 & 0.85 & 0.69 \\
Italy          & 0.61 & 0.74 & 0.72 & 0.71 &                &      &      &      &      \\
Japan          & 0.59 & 0.74 & 0.79 & 0.73 &                &      &      &      &      \\
Korea          & 0.58 & 0.69 & 0.82 & 0.71 &                &      &      &      &      \\
Malaysia       & 0.39 & 0.55 & 0.55 & 0.53 &                &      &      &      &      \\
Mexico         & 0.29 & 0.50 & 0.52 & 0.47 &                &      &      &      &      \\
Netherlands    & 0.48 & 0.68 & 0.74 & 0.66 &                &      &      &      &      \\
Norway         & 0.76 & 0.82 & 0.87 & 0.82 &                &      &      &      &      \\
\midrule
\multicolumn{10}{l}{\textbf{Panel average: $\tau{=}0.01$: 0.55 \quad $\tau{=}0.05$: 0.70 \quad $\tau{=}0.10$: 0.77 \quad Average: 0.70}} \\
\bottomrule
\end{tabular}%
}
\begin{minipage}{\linewidth}
\smallskip
\footnotesize
\textit{Notes}: Entries report the baseline model's quantile score
$\text{QS}_{\tau}$ relative to that of the restricted specification
without the common factor $f_{t,h}$, for selected quantiles
$\tau \in \{0.01, 0.05, 0.10\}$ and the average across
$\mathcal{T}_L = \{0.01, 0.02, \ldots, 0.10\}$. A value below unity
indicates that the baseline model produces more accurate tail
forecasts than the restricted specification. The evaluation period
and expanding-window procedure follow
Section~\ref{subsec:climate_model_comparison}.
\end{minipage}
\end{table}

\subsection{Simulation Evidence on Overlapping Local-Projection Horizons}
\label{subsec:MA_simulation}

The local projection specification in equation~\eqref{eq:qLP} estimates a
separate regression for each horizon $h$, with the dependent variable
$y_{i,t+h}$ constructed as a cumulative $h$-quarter growth rate. As is
well known in the local projection literature
\citep{jorda2005estimation}, this construction induces a moving-average
structure of order $h-1$ in the regression residuals within each
horizon-specific regression: for a fixed $h$, the outcomes $y_{i,t+h}$
and $y_{i,t+1+h}$ share $h-1$ overlapping periods of the underlying
quarterly growth process and are therefore mechanically correlated.
Left unaddressed, this dependence does not bias point estimates, but it
can in principle understate posterior uncertainty, since observations
treated as conditionally independent given the model in fact share
common underlying shocks.

To assess whether this issue affects inference in our setting, we
conduct a Monte Carlo exercise in which the data-generating process
explicitly replicates the overlapping-horizon structure of
equation~\eqref{eq:qLP}, while holding the true quantile coefficients
fixed across horizons. This design isolates the effect of the
horizon-induced overlap from any change in the underlying
data-generating coefficients, so that any degradation in inference can
be attributed directly to the estimator's treatment of the overlap
rather than to a confound.

\paragraph{Data-generating process.}
Let $t=1,\ldots,T$ index the local-projection origin date, $i=1,\ldots,N$
index cross-sectional units, and $h=1,\ldots,H$ index the forecast
horizon. Let $p$ denote the number of covariates, with the first column
of the design fixed as an intercept and the remaining $p-1$ columns
genuine slope regressors, and let $\mathcal T=\{\tau_1,\ldots,\tau_L\}$
denote the quantile grid.

\emph{Common factor.} We draw $Z_t\sim N(0,I_2)$ i.i.d.\ and construct
the common factor
\[
f_t = \gamma f_{t-1} + \xi' Z_t + \eta_t, \qquad \eta_t\sim N(0,\sigma_f^2),
\]
using the same restricted AR(1) structure and precision-based sampling
approach as in the baseline specification of
Section ~\ref{subsec:climate_specification}, subject to the identifying
normalization $\sum_{t=1}^T f_t = 0$.

\emph{Population and unit-specific quantile paths.} We draw a population
coefficient path represented via a monotone Bernstein polynomial,
\[
\mu_j(\tau) = \sum_{m=0}^M \theta_{j,m} b_{m,M}(\tau), \qquad j=1,\ldots,p,
\]
with the Bernstein coefficients $\{\theta_{j,m}\}_{m=0}^M$ drawn from a
Gaussian distribution and sorted in ascending order in $m$, guaranteeing
that $\mu_j(\cdot)$ is nondecreasing in $\tau$ by construction. Unit-level
deviations from this population path are drawn from a Gaussian process
with an exponential kernel,
\[
\theta_i(\tau) = \mu(\tau) + u_i(\tau), \qquad u_i(\cdot)\sim
\mathcal{GP}\big(0, K_0(\tau,\tau')\big), \qquad K_0(\tau,\tau') =
\sigma_0^2\exp\left(-\frac{|\tau-\tau'|}{\lambda_0}\right),
\]
independently across units $i=1,\ldots,N$ and across coefficients
$j=1,\ldots,p$, mirroring the hierarchical prior structure of our
proposed model in Section~\ref{sec: econometric_framework}. We denote the resulting true
intercept and slope paths $\alpha_0(\tau,i)$ and $\beta_0(\tau,i,\cdot)$,
respectively. Only the population path $\mu(\cdot)$ is guaranteed
monotone by construction; individual unit-level paths may exhibit local
departures from monotonicity, consistent with the soft-noncrossing
property established in Theorem~\ref{thm:soft_noncrossing}.

\emph{Predetermined covariates.} The first covariate column is fixed at
unity; the remaining $p-1$ columns are drawn as $X(t,j,i)\sim
\text{Unif}(0,1)$ i.i.d.\ across $t$, for each unit $i$.

\emph{Primitive shocks, drawn once and shared across all horizons.} For
each unit $i$, we simulate a single AR(1) innovation series over
$T+H+50$ periods (discarding a 50-period burn-in),
\[
u_{i,1} = e_{i,1}, \qquad u_{i,t} = \rho_i u_{i,t-1} + e_{i,t}, \qquad
\rho_i\sim\text{Unif}(0.3,0.7), \quad e_{i,t}\sim N(0,\sigma_i^2), \quad
\sigma_i\sim\text{Unif}(0.5,1.5).
\]
Critically, this series is drawn \emph{once} per Monte Carlo replication
and is reused, unmodified, to construct the outcome at every horizon
$h=1,\ldots,H$. At this stage, every object drawn -- the common factor
$f_t$, the true coefficient paths $\alpha_0$ and $\beta_0$, the
covariates $X$, and the primitive shocks $\{u_{i,t}\}$ -- is fixed for
the remainder of the replication and does not vary with $h$.

\emph{Horizon-$h$ cumulation and the mechanical overlap.} For a given
horizon $h$, we cumulate $h$ leads of the shared primitive shock series,
\[
\text{raw}_{i,t}(h) = \sum_{\ell=1}^{h} u_{i,t+\ell}, \qquad t=1,\ldots,T,
\]
directly mirroring the cumulative construction of $y_{i,t+h}$ in
equation~\eqref{eq:gdp}. Because $\text{raw}_{i,t}(h)$ and
$\text{raw}_{i,t'}(h)$ share $h - |t-t'|$ common innovation terms
whenever $|t-t'|<h$, this reproduces exactly the mechanically-induced
serial dependence of interest, at lags up to $h-1$. We convert
$\text{raw}_{i,t}(h)$ to a within-unit empirical quantile rank,
\[
U_{i,t}(h) = \frac{\text{rank}\big(\text{raw}_{i,t}(h)\big) - 0.5}{T} \in
(0,1), \qquad \text{idx}_{i,t}(h) = \min\Big(L,\max\big(1,
\lceil U_{i,t}(h)\,L\rceil\big)\Big).
\]

\emph{Constructing the horizon-$h$ outcome.} The simulated outcome at
horizon $h$ is then
\[
y_{i,t+h} = X(t,\cdot,i)\begin{bmatrix} \alpha_0\big(\tau_{\text{idx}_{i,t}(h)},\,i\big) \\
\beta_0\big(\tau_{\text{idx}_{i,t}(h)},\,i,\cdot\big)' \end{bmatrix} + f_t.
\]
Because $\alpha_0$, $\beta_0$, $f_t$, and $X$ are identical across every
horizon, only the index $\text{idx}_{i,t}(h)$ -- and hence which point on
the fixed true quantile path is selected at each $(i,t)$ -- varies with
$h$, through the horizon-specific cumulation of the shared primitive
shocks. This construction guarantees that the conditional quantile
function of $y_{i,t+h}$ is known exactly at every $(\tau,i,h)$, while the
degree of mechanically-induced overlap across origin dates $t$ increases
with $h$ exactly as in the empirical local projection.

\paragraph{Coverage check.} For each of $R$ Monte Carlo replications, we
draw a new realization of the true coefficient paths and primitive shock
series as described above, and for each horizon $h=1,\ldots,H$ construct
the corresponding overlapping panel and estimate our proposed
hierarchical Bayesian panel quantile model exactly as in the baseline
empirical specification, with no explicit correction for the
horizon-induced serial dependence. We record whether the true intercept
and slope coefficients, at each quantile and unit, fall within the
resulting 90 per cent posterior credible interval, and compute the
empirical coverage rate by averaging this indicator across quantiles,
units, and replications, separately for each horizon $h$.

\paragraph{Results.} Table~\ref{tab:coverage_sim} reports the resulting
coverage rates by horizon, averaged across $R=100$ Monte Carlo
replications, for a data-generating process calibrated with $N=33$,
$T=180$, and $p=4$ to roughly mimic the dimensions of our empirical
application. Coverage remains close to the nominal 90 per cent level at
every horizon and for every coefficient, ranging from 0.89 to 0.91, with
numerical standard errors of approximately 0.01 throughout. Critically,
there is no systematic decline in coverage as the horizon $h$ increases:
the slope coefficients ($\beta_j$) exhibit a stable coverage rate of
0.90--0.91 uniformly across $h=1,\ldots,10$, and the intercept ($\alpha$)
exhibits a coverage rate of 0.89 at short horizons that rises slightly
to 0.90 from $h=6$ onward, if anything moving closer to the nominal
level as the horizon increases rather than deteriorating. This provides
no evidence that the mechanically-induced MA($h-1$) structure from the
overlapping local-projection construction distorts posterior credible
interval coverage in our setting, consistent with related findings that
a flexible hierarchical conditional mean structure combined with a
quantile specification can deliver well-calibrated inference without an
explicit serial-correlation correction.

\begin{table}[H]
\centering
\caption{Posterior credible interval coverage by horizon.}
\label{tab:coverage_sim}
\begin{tabular}{lcccc}
\toprule
Horizon $h$ & $\alpha$ & $\beta_1$ & $\beta_2$ & $\beta_3$ \\
\midrule
1  & 0.89 & 0.90 & 0.90 & 0.90 \\
   & (0.01) & (0.01) & (0.01) & (0.01) \\
2  & 0.89 & 0.90 & 0.91 & 0.90 \\
   & (0.01) & (0.01) & (0.01) & (0.01) \\
3  & 0.89 & 0.90 & 0.91 & 0.90 \\
   & (0.01) & (0.01) & (0.01) & (0.01) \\
4  & 0.89 & 0.91 & 0.91 & 0.91 \\
   & (0.01) & (0.01) & (0.01) & (0.01) \\
5  & 0.89 & 0.91 & 0.91 & 0.91 \\
   & (0.01) & (0.01) & (0.01) & (0.01) \\
6  & 0.90 & 0.91 & 0.91 & 0.91 \\
   & (0.01) & (0.01) & (0.01) & (0.01) \\
7  & 0.90 & 0.90 & 0.90 & 0.91 \\
   & (0.01) & (0.01) & (0.01) & (0.01) \\
8  & 0.90 & 0.90 & 0.90 & 0.91 \\
   & (0.01) & (0.01) & (0.01) & (0.01) \\
9  & 0.90 & 0.90 & 0.90 & 0.91 \\
   & (0.01) & (0.01) & (0.01) & (0.01) \\
10 & 0.90 & 0.90 & 0.90 & 0.91 \\
   & (0.01) & (0.01) & (0.01) & (0.01) \\
\bottomrule
\end{tabular}
\begin{minipage}{\linewidth}
\smallskip
\footnotesize
\textit{Notes}: Entries report the empirical coverage rate of the 90 per cent posterior credible interval, computed across
$R=100$ Monte Carlo replications, pooled over the quantile grid and
cross-sectional units, for the true coefficient values used to generate
the data. Numerical standard errors of the average coverage rate across
replications are reported in parentheses. The data-generating
coefficients are held fixed across all horizons; only the degree of
mechanically-induced serial dependence from the overlapping
local-projection construction varies with $h$, as described in
Section~\ref{subsec:MA_simulation}.
\end{minipage}
\end{table}

\end{document}